\pdfoutput=1
\documentclass[american,aps,pra,reprint,superscriptaddress]{revtex4-1}
\usepackage[T1]{fontenc}
\usepackage[latin9]{inputenc}
\usepackage{color}
\usepackage{babel}
\usepackage{amsmath}
\usepackage{amsthm}
\usepackage{graphicx}
\usepackage{amssymb}
\usepackage{times}
\usepackage{txfonts}
\usepackage{tikz}
\usepackage{hyperref}
\usepackage{qcircuit}
\usepackage{bbm}

\newtheorem{theorem}{Theorem}

\newtheorem{lemma}[theorem]{Lemma}
\newtheorem{corollary}[theorem]{Corollary}

\newcommand{\ZZ}{\mathbb{Z}}
\newcommand{\RR}{\mathbb{R}}
\newcommand{\CC}{\mathbb{C}}

\newcommand{\NN}{\mathbb{N}}

\newcommand{\mc}[1]{\mathcal{#1}}
\newcommand{\bra}[1]{\langle #1|}
\newcommand{\ket}[1]{|#1\rangle}
\newcommand{\braket}[2]{\langle #1|#2\rangle}
\newcommand{\ketbra}[2]{| #1 \rangle \langle #2 |}

\newcommand{\proj}[1]{\vert #1\rangle\!\langle#1 \vert}

\def\id{{\mathbb{I}}}
\newcommand{\norm}[1]{\left\Vert #1 \right\Vert}
\def\one{\mathbf{1}}

\newcommand{\Tr}{\operatorname{tr}}
\newcommand{\tr}{\Tr}

\DeclareMathOperator*{\argmin}{arg\,min}

\begin{document}
\author{P. Boes, H. Wilming, R. Gallego, J. Eisert}
\newcommand{\fu}{Dahlem Center for Complex Quantum Systems, Freie Universit{\"a}t Berlin, 14195 Berlin, Germany}
\affiliation{\fu}

\title{Catalytic quantum randomness}
\date{\today}

\begin{abstract}
Randomness is a defining element of mixing processes in nature and an essential ingredient to
many protocols in quantum information. In this work, we 
investigate how much randomness is required to transform a given quantum state into another one. Specifically, 
we ask whether there is a gap between the power of a classical source of randomness compared to that of a quantum one. We provide a complete answer to these questions, by identifying provably
 optimal protocols for both classical and quantum sources of randomness, based on a dephasing construction. 
 We find that in order to implement any noisy transition on a $d$-dimensional quantum system it is necessary and sufficient to have a quantum source of randomness of dimension $\sqrt{d}$ or a classical one of dimension $d$. 
 Interestingly, coherences provided by quantum states in a source of randomness offer
 a quadratic advantage.
 The process we construct has the additional features to be robust and 
 catalytic, i.e., 
 the source of randomness can be re-used. 
Building upon this formal framework, we illustrate that this dephasing construction can serve as a useful primitive in both equilibration and quantum information theory: We discuss
 applications describing the smallest measurement device, capturing the smallest
 equilibrating environment allowed by quantum mechanics, or forming the basis for a cryptographic
 private quantum channel. We complement the exact analysis with a discussion of approximate protocols based on quantum expanders deriving from discrete Weyl
systems. This gives rise to equilibrating environments of remarkably small dimension. Our 
results highlight the curious feature of randomness that residual correlations and dimension 
can be traded against each other.
\end{abstract}

\maketitle
\section{Introduction}

Randomness is a central concept and resource in various fields of research
in computer science, information theory and physics, in both the classical and the quantum realm.
It is an ingredient to (quantum) algorithm design, a core element in 
coding and communication protocols, and plays a central role in fundamental aspects
of statistical mechanics. In the quantum context, randomness is also increasingly
being seen as a valuable resource. {A natural question that arises in this
context is then how much of it is required to implement a given physical process on a quantum system. 
Another important question is to what extent the required amount of randomness differs depending on whether an \emph{implicit} or an \emph{explicit} model of randomness is employed. Here, an implicit model of randomness considers the \emph{source of randomness (SoR)} as a black box that provides coin flips, while an explicit model takes into account the fact that, fundamentally, all systems including the ones provided by the SoR are quantum systems, and hence models the randomness as a quantum state. 

In this work, we give a complete answer to both of the above questions. We provide, for both the implicit and explicit model, optimal and tight bounds on the amount of randomness required to implement physical processes on quantum systems. Moreover, we show a strict separation between the above models, in the sense that every physical process can be implemented in the explicit model by using only half the amount of randomness that is required in the implicit model.}

Specifically, we use a model of noisy processes---processes that require randomness--- known as \emph{noisy operations}~\cite{Gour2015Resource}. 
We study the minimal amount of noise required to implement a large variety of noisy processes and construct protocols that saturate the lower bounds imposed by quantum mechanics. These processes include 
\emph{dephasing and equilibration}~\cite{LongReview,Linden2012}, \emph{decoherence}~\cite{Zurek2003,Zeh}, 
the \emph{implementation of measurements}~\cite{HolevoBook,Nielsen2000,Zeh}, any \emph{transition} between two quantum states that requires randomness~\cite{Gour2015Resource} as well as the novel construction of \emph{private quantum channels}~\cite{Ambainis2000,Boykin2003Optimal}. 

It is an important aspect of our work that, by virtue of an explicit model, these saturated lower bounds also translate into bounds on the \emph{physical size} of an SoR. This insight allows us to
construct, for particular processes, the \emph{smallest decohering environment} 
or measurement device compatible with quantum mechanics~\cite{Zurek2003}. Put in a different language, 
it provides an understanding of the \emph{smallest equilibrating 
environment}~\cite{LongReview} possible. The surprisingly small size that suffices for an environment to be equilibrating challenges the commonly held view that such decohering baths should necessarily feature 
a large
dimension.

A further notable feature of the protocols that we construct is that they are \emph{catalytic}: 
The same unit of randomness can be \emph{re-used} for different processes~\cite{Mueller2017}.
It is also 
\emph{robust}, in the sense that we do not require perfect control in either the states prepared by the SoR or the timing of the process,
and further \emph{recurrent}, in the sense that, for large system dimension $d$, continuous time versions of our noisy processes maintain a state close to the desired final state for times $\tau \propto \sqrt{d}$, at which point the systems recurs to the initial state. 

\section{Classical versus quantum noise}\label{sec:classicalvsquantum}

Let us begin with discussing in more detail the difference between classical and quantum uses of randomness.
Consider initial and final (mixed) states $\rho,\rho'$ on 
a Hilbert space $\mc H_S$ of dimension ${\rm \dim}( \mc H_S)=d$. 
We are concerned with the possibility of implementing a transition $\mc E(\rho)=\rho'$, where 
$\mc E$ represents a noisy process. 
There exist different ways of modeling the maps $\mc E$ which we now explain in detail.

In a classical, implicit model of the SoR one assumes a discrete random variable $J$ that is uniformly distributed over $m$ possible values. 
Depending on the value of $j$ one implements a given unitary transformation $U_j$, which gives rise to the operations
\begin{align}\label{eq:def:classicalnoise}
\mc E_{\rm C}^m(\:\cdot \:)=\frac{1}{m} \sum_{i=1}^{m} U_i \cdot U_i^{\dagger}.
\end{align}
If there exist $\mc E_{\rm C}^m$ so that a transition is possible, we simply denote it by $\rho  \overset{m}{\to}_C  \rho'$. In constrast, in an explicit quantum model, the SoR is a quantum system $R$ in the maximally mixed state of dimension $m$, which we denote by $\id_m \coloneqq \frac{1}{m} \mathbbm{1}$, with $\mathbbm{1}$ being the identity matrix.
In this model, noisy processes are any effect of a unitary joint evolution of the compound,
\begin{align}\label{eq:def:quantumnoise}
\mc E_{\rm Q}^m(\:\cdot \:)=  \tr_R[U(\:\cdot \: \otimes \id_m)U^\dagger].
\end{align}
As in the classical case, we write $\rho \overset{m}{\to}_Q \rho'$ whenever the transition is possible. 

The set of transitions that can be implemented with both classical and quantum noise coincide if the amount of 
noise --- quantified by the dimension $m$ --- is unbounded. In this case we have
\begin{align} \label{eq:equiv_for_infty}
\rho \overset{\infty}{\to}_C \rho' \Leftrightarrow \rho \overset{\infty}{\to}_Q \rho' \Leftrightarrow \rho \succ \rho'
\end{align}
where we use the symbol ``$\succ$'' to indicate that $\rho$ majorizes $\rho'$~\cite{Marshall2011Inequalities}. 
The set of transitions $\rho \overset{\infty}{\to}_Q \rho'$ have been extensively studied as \emph{noisy operations}~\cite{Gour2015Resource}, where the noise is treated as a free resource and the main concern is to study the possible transitions with unbounded $m$. 
In contrast, here we are concerned with treating noise as a valuable resource and focus on the following question: 
What is the minimal amount of noise---quantified by $m$---that serves to implement any possible transition between pairs of $d$-dimensional quantum states fulfilling $\rho \succ \rho'$? 
We denote these minimal values of $d$ for the classical and quantum case by $m^*_C(d)$ and $m^*_Q(d)$, respectively.

At first glance, one might suspect that $m^*_C(d)=m^*_Q(d)$, with quantum noise offering no advantage over its classical counterpart. 
That intuition comes from the fact that, although one writes a full quantum description in~\eqref{eq:def:quantumnoise}, the state of $R$, given by $\id_m$, is nevertheless a quasi-classical state. 
Hence it seems reasonable that it could be recast as a classical variable, similarly as in~\eqref{eq:def:classicalnoise}. 
However, treating the noise as a quantum state allows one to access its quantum degrees of freedom, for example to create entanglement between the $S$ and $R$. 
In other words, one could in principle use quantum correlations to make a more efficient use of the noise yielding $m^*_C(d)>m^*_Q(d)$. 

One of the main results of this work is to show that there is indeed a gap between the classical and quantum case. 
We find that $m^*_C(d)=d \:{>}\:\lceil d^{1/2}\rceil=m^*_Q(d)$ and more importantly, we construct protocols that saturate those bounds. 
In this way, we provide protocols that use the noise optimally for a large variety of tasks. 
These protocols also have a number of useful properties such as allowing one to re-use the noise or being robust under different classes of imperfections.
In the subsequent section, we present the key lemma to construct such optimal protocols and then turn to discuss applications and properties in Section~\ref{sec:applications}.

\section{An optimal dephasing map} 
\label{sec:an_optimal_dephasing_map}

For any state transition $\rho\to \rho'$ that is possible under either quantum or classical noisy processes, there exists a corresponding map $\mc E(\rho) = \rho'$ such that
\begin{align} \label{eq:decomp_channel}
    \mc E(\cdot) = \mc U' \circ \pi_{A} \circ \mc U(\cdot).
\end{align}
Here $\mc U', \mc U$ are unitary channels that depend on $\rho$ and $\rho'$. 
The map $\pi_{A}$ is the dephasing map in a fixed orthonormal basis $A = \{\ket{i}\}_{i=1}^d$, defined as 
\begin{align}
    \bra{i}\pi_A(\rho) \ket{j} = \bra{i}\rho\ket{j} \delta_{i,j},
\end{align}
with $\delta_{i,j}$ being the Kronecker delta. 
This follows from the Schur-Horn-Theorem~\cite{Horn1954Doubly} together with~\eqref{eq:equiv_for_infty} and was used to bound the required randomness for noisy processes already in Ref.\ \cite{Scharlau2016Quantum}. 
Since the unitary channels $\mc U', \mc U$ do not require the use of any SoR by definition, we see from~\eqref{eq:decomp_channel} that noise is required only for the implementation of the dephasing map $\pi_A$. 
In turn,~\eqref{eq:decomp_channel} implies that whether $\mc E$ represents a quantum noisy process or a classical one, depends only on the particular implementation of this dephasing map: 
Any construction of $\pi_A$ in the form of~\eqref{eq:def:quantumnoise} with $m$-dimensional SoR 
implies also that $\mc{E}$ is a map $\mc E_Q^m$, while any construction of it in the form of~\eqref{eq:def:classicalnoise} implies that $\mc{E}$ is of the form $\mc E_C^m$.

Understanding the amount of randomness required to implement the dephasing map therefore is key to understanding the amount of randomness required to implement any noisy process. The following lemma provides a protocol implementing a dephasing map in any basis, using an explicit model model of noise and requiring  a SoR of dimension $m=\lceil d^{1/2} \rceil$.
\begin{lemma}[Catalytic quantum dephasing]\label{lem:dephasing_basic}
 For any integer $d$ and basis $A$ there exists a unitary $U$ so that
 \begin{align}
	\label{eq:dephasingcondition1}     &\Tr_R[U \: (\cdot \: \otimes\: \id_{\lceil d^{1/2}\rceil} )\: U^\dagger] = \pi_A(\cdot),\\
	 \label{eq:dephasingcondition2}   &\Tr_S[U (\rho   \otimes  \id_{\lceil d^{1/2}\rceil} ) U^\dagger] = \id_{\lceil d^{1/2}\rceil}\:\: \forall \rho.
 \end{align}
\end{lemma}

\begin{proof}
Assume first that $ \sqrt{d} = m \in \NN$. 
Now, let $\{U_i\}$ be a unitary operator basis for $\mc B(\mc H_R)$, that is, a collection of $m^2=d$ unitary operators $U_i \in \mc B(\mc H_R)$ such that 
\begin{align} \label{eq:prop_unitary_basis}
     \frac{1}{m}\Tr(U_i U_j^\dagger) &= \delta_{i,j}
\end{align}
for all $i,j$.
Such a basis exists for every $m$~\cite{Schwinger1960,Werner2001All}. 
We now define the unitary
\begin{align} \label{eq:dephasing_unitary_basic}
     U = \sum_{i=1}^d \proj{i} \otimes U_i,
\end{align}
where the $\{\ket{i}\}$ are elements of the basis $A$ in which we intend to pinch. 
Then, for any density matrix $\rho$ on $\mc H_S$, 
\begin{align}
    \tr_R[U (\rho \otimes \id_m )U^\dagger] &= \sum_{i,j} \proj{i}\rho \proj{j} \frac{1}{m}\Tr(U_i U^\dagger_j) \\
	    &= \sum_{i,j} \proj{i}\rho \proj{j} \delta_{i,j} = \pi_A(\rho).
\end{align}
Lastly, note that Eq.\ ~\eqref{eq:dephasingcondition2} follows simply by
\begin{align}\label{eq:map_on_noise}
 \tr_S[U (\rho  \otimes \id_m ) U^\dagger] &=\sum_i \bra{i}\rho \ket{i} U_i  \id_m U_i^\dagger=\id_ m .
\end{align}
In the case where $\sqrt{d}$ is not an integer, we can use the same construction with a source of randomness of dimension $m =\lceil d^{1/2} \rceil$ by simply not exhausting all possible $m^2$ possible unitaries $U_i$ on $R$. 
\end{proof}

The protocol of Lemma~\ref{lem:dephasing_basic} is optimal, in the sense that it is impossible to implement the dephasing map with $m<\lceil d^{1/2}\rceil$. 
This can be seen by noting that for any basis $A$ one can always choose an initial pure state $\rho$ so that $\pi_A(\rho)=\id_d$. 
Using the preservation of the von Neumann-entropy under unitaries and the Lieb-Araki triangle inequality one finds that $m\geq \sqrt{d}$ (see Appendix~\ref{sec:necessary}). 
This implementation of the dephasing map compares with the best value known to date of $m=d$, proven in 
Ref.~\cite{Scharlau2016Quantum}, whose implementation can in fact be shown to correspond to a classical noisy operation of the form~\eqref{eq:def:classicalnoise} as we will see later.

\subsection{Catalyticity} 
\label{sub:catalyticity}

Equation~\eqref{eq:dephasingcondition2} states that the dephasing operation defined in Lemma~\ref{lem:dephasing_basic} leaves the state of $R$ invariant, or in other words, that the noise is catalytic~\cite{JonathanPlenio, Mueller2017,Ng14,PhysRevLett.85.437}.
This property has numerous useful applications.
For instance, an immediate corollary of the lemma is that one can locally dephase an arbitrarily large number of uncorrelated systems, each of them of dimension at most $d$, by using a single noise system $R$ of dimension $\lceil d^{1/2} \rceil$. 
More formally, we have that for any set of states $\{\rho^i\}_{i=1}^{N}$ there exists a unitary $U$ so that
\begin{align}\label{eq:local_dephasing}
\tr_R[U (\rho^i_{S_1} \otimes \cdots \otimes  \rho^i_{S_N}\otimes \id_{\lceil d^{1/2}\rceil}) \: U^\dagger]=  \rho'_{S_1,\ldots,S_N}
\end{align}
where $\rho'_{S_i}=\pi_{A_i}(\rho^i_{S_i})$. 
This follows by simply iterating the unitaries of Lemma~\ref{lem:dephasing_basic} with all the subsystems and re-using the noise as illustrated in the top of Fig.~\ref{fig:dephasing}. 
In contrast, if the noise would not have the property of being catalyticm, then it would be necessary to employ a new mixed state for each of the subsystems, in which case an amount of randomness proportional to $N$ would be required. (bottom of Fig.~\ref{fig:dephasing}). It is important to note, however, that reusing the randomness comes at the cost of correlating the subsystems amongst each other. Hence, if a protocol requires for the individual systems to remain uncorrelated, one still has to resort to a scheme whose required randomness scales linearly with the number of subsystems.

\begin{figure}[tb]
\includegraphics[width=0.4\textwidth]{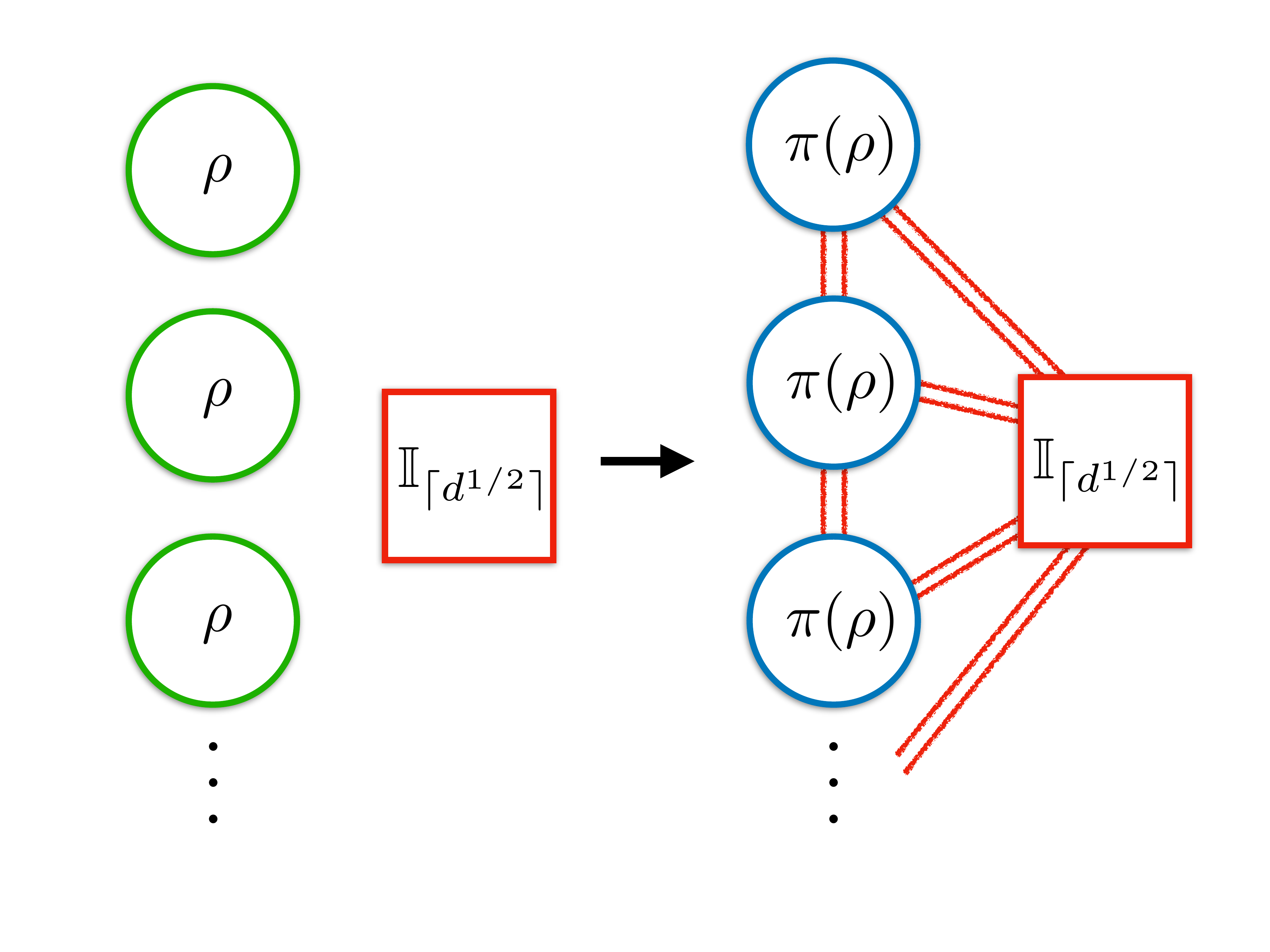}
\includegraphics[width=0.4\textwidth]{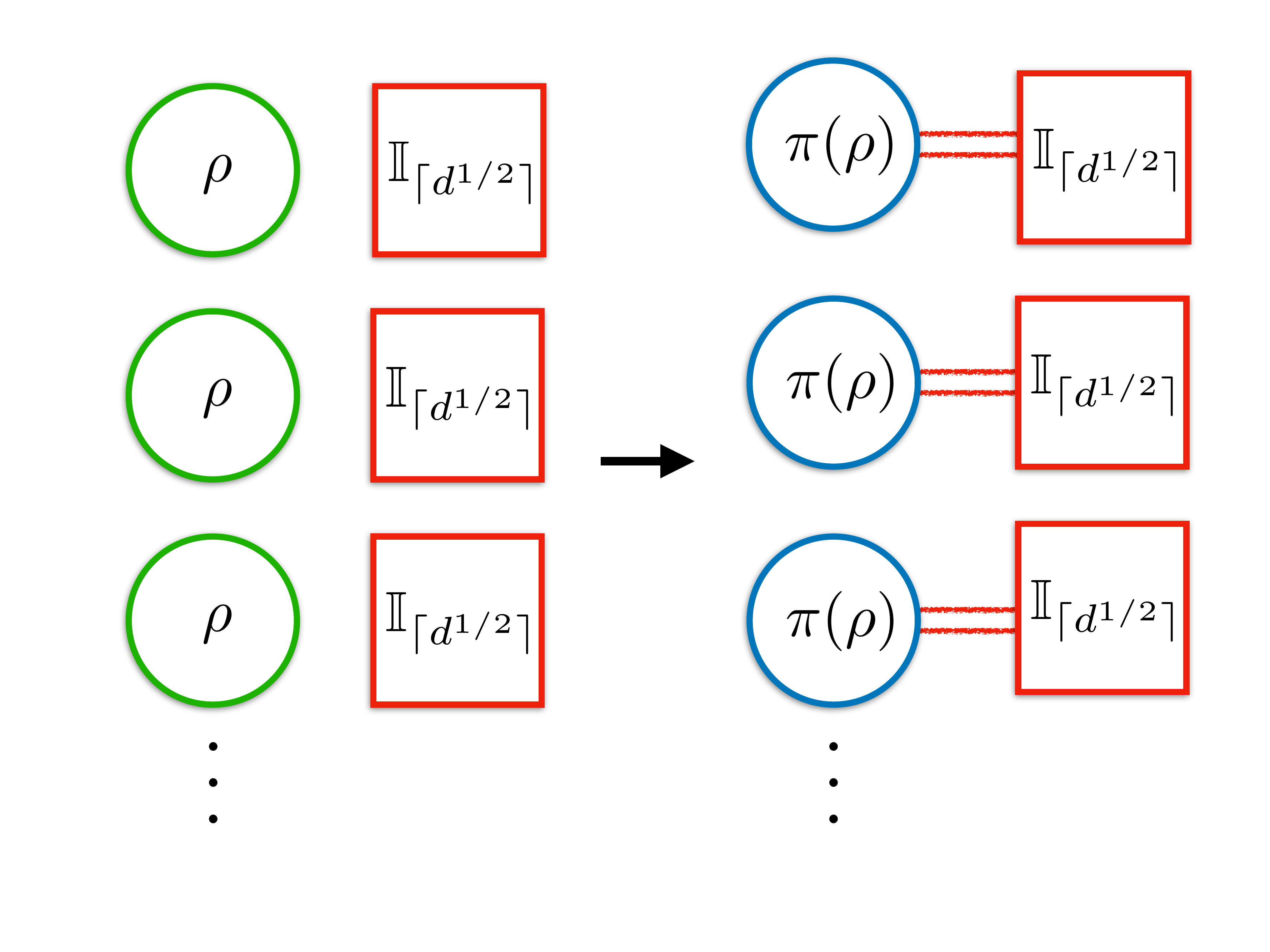}
\caption{Two possible ways of dephasing and the resulting correlation structure. Top: A sequence of systems in state $\rho$ is dephased using a single state of randomness, with correlations being established between all systems involved. The local margins of the resulting global state \eqref{eq:local_dephasing} are the dephased initial states. Bottom: In order to avoid correlations between the systems, one can instead use additional and unused randomness.} 
\label{fig:dephasing}
\end{figure}

As sketched already, dephasing can be related to many processes that require noise, both in engineered
as well as in equilibrating natural quantum processes. In the remainder of this work, we discuss and present applications of Lemma~\ref{lem:dephasing_basic} to these processes.

\section{Applications} 
\label{sec:applications}

\subsection{Minimal noise for state transitions} 
\label{sub:first_application_noisy_operations}

As a first application, we prove the tight bounds for noisy operations presented in Section~\ref{sec:classicalvsquantum}. 
Formally, given a Hilbert space $\mc H_S$ with $\dim(\mc H_S)=d$, we define the minimal noise for the classical and quantum case as
\begin{align}
\label{eq:def:minimalclassical}m^*_C(d)\coloneqq &\argmin_m \: \rho \overset{m}{\to}_C \rho' \:\:\:\forall \rho,\rho' \in \mc B(\mc H_S)\: | \: \rho \succ \rho' ,\\
\label{eq:def:minimalquantum}m^*_Q(d)\coloneqq &\argmin_m \: \rho \overset{m}{\to}_Q \rho' \:\:\:\forall \rho,\rho' \in \mc B(\mc H_S)
\: | \: \rho \succ \rho'.
\end{align}
In the following lemma we find the values of the above quantities, thus providing the smallest SoR that suffices to perform any transition between two states $\rho \succ \rho'$. Note, however, that it is possible for {particular transitions to require even less randomness or none at all.} 
\begin{lemma}[Optimal {source of randomness} for state transitions]\label{lemma:optimaldimensions}
Any state transition of a $d$-dimensional system that is possible under noisy processes, in the sense of~\eqref{eq:def:minimalclassical} and~\eqref{eq:def:minimalquantum}, can be implemented using an amount of classical and quantum noise given by
\begin{align}
\label{eq:optimalclassical}&m^*_C(d)=d,\\
\label{eq:optimalquantum}&m^*_Q(d)=\lceil d^{1/2} \rceil.
\end{align}
\end{lemma}
\begin{proof}
Here, we only prove that the above values are sufficient. For the corresponding necessary conditions (and $\epsilon$-approximate versions of the above) see Appendix~\ref{sec:necessary}.
Eq.~\eqref{eq:optimalquantum} follows from combining~\eqref{eq:decomp_channel} with the dephasing construction in Lemma~\ref{lem:dephasing_basic}. 
To see~\eqref{eq:optimalclassical}, consider the {unitary
\begin{align}
	V= \sum_{i=1}^d \proj{i}_S \otimes X^i_R,
 \end{align}
 where $X$ is the generalized Pauli matrix defined as 
 \begin{align}\label{X}
 	X \ket{i} = \ket{(i + 1) \text{ mod } d}.
  \end{align} 
 As shown in Ref.~\cite{Scharlau2016Quantum}, this unitary implements the dephasing map
 \begin{align}
	 \tr_R(V (\rho \otimes \id_d )V^\dagger) &= \frac{1}{d} \sum_{i,j}\bra{i}\rho\ket{j}\ketbra{i}{j} \tr(X^{i-j}) &=  \pi_A(\rho).
 \end{align}
 $V$ is the local Fourier transform of a unitary leading to a channel of the form~\eqref{eq:def:classicalnoise}: 
 there exists a unitary $F$ and a basis $\{\ket{ \tilde j}=F^\dagger \ket{j}\}$ such that
 \begin{align}
	 \tilde{V}\coloneqq (\mathbbm{1} \otimes F) V (\mathbbm{1}\otimes F^\dagger) =\sum_{j=1}^d  Z^j \otimes \proj{\tilde{j}}.
 \end{align}
Here, 
 \begin{align}\label{Z}
Z=\sum_j \omega_d^j \proj{j} 
 \end{align}
 is the generalized Pauli matrix conjugate to $X$ and $\omega_d$ the $d$-th root of unity.
Since the maximally mixed state is unitarily invariant, $\tilde{V}$ implements the dephasing map and its action on the system $S$ can be represented as
\begin{align}
	\rho \mapsto \tr_R(\tilde V (\rho\otimes \id_d) {\tilde V}^\dagger) = \frac{1}{d}\sum_{j=1}^d Z^j \rho Z^{-j}.
\end{align}
Thus the dephasing map can be implemented with a classical SoR of dimension $d$. }
\end{proof}
This lemma proves a conjecture in Ref.~\cite{Scharlau2016Quantum}, where the possibility of strengthening their bound $m^*_Q(d)=d$ to the present one was already raised.

In complete analogy to the discussion in Section~\ref{sub:catalyticity} and Fig.~\ref{fig:dephasing}, we can also use the catalytic properties of the source of randomness to implement state transitions locally from an initially uncorrelated state and using a fixed-size source of randomness. More concretely, let $\{\rho^i\}_{i=1}^N$ and $\{\sigma^i\}_{i=1}^N$ be $d$-dimensional quantum states such that $\rho^i \succ \sigma^i$ for all $i=1,\ldots,N$. Then there exists a unitary $U$ such that
\begin{align}\label{eq:local_transition}
\tr_R[U (\rho^1_{S_1} \otimes \cdots \otimes  \rho^N_{S_N}\otimes \id_{\lceil d^{1/2}\rceil}) \: U^\dagger]=  \rho'_{S_1,\ldots,S_N},
\end{align}
with $\rho'_{S_i} = \sigma^i$. To see this, we recall from the discussion in section~\ref{sub:catalyticity} that the transition $\rho^i \rightarrow \sigma^i$ can be implemented composing unitary channels and dephasing maps. Hence, $\mc{E}(\rho^{1}_{S_1} \otimes \cdots \otimes \rho^{1}_{S_1} ) = \sigma^{1}_{S_1} \otimes \cdots \otimes \sigma^{1}_{S_1}$ with
\begin{align}\label{eq:local_shur_horn}
\mc{E} = \bigotimes_{i=1}^N \mc{U'}_{S_i} \circ \bigotimes_{i=1}^N\pi_{A_1}  \circ \bigotimes_{i=1}^N\mc{U}_{S_1}
.
\end{align}
Now, using Eq.\ 
\eqref{eq:local_dephasing} we see that it is possible to dephase locally ---that is, perform locally the same transition as the one implemented by the second map on the r.h.s. of \eqref{eq:local_shur_horn}--- using a single source of randomness of dimension $\lceil d^{1/2}\rceil$, at the cost of creating correlations between the subsystems. Hence, composing the local unitaries with the local dephasing of \eqref{eq:local_dephasing}, we obtain a map that locally implements the same transition as $\mc{E}$, as captured by \eqref{eq:local_transition}.

\subsection{Smallest possible decohering environment and measurement device} 
\label{ssub:measurements}
A further application of our results is to the physical mechanism of decoherence and implementing a measurement in quantum mechanics, which can indeed be seen as a special case of a noisy operation, since it requires randomness. Both applications follow from the fact that a quantum source of randomness can be seen as half of a maximally entangled system.  

It is useful to first discuss decoherence. 
To do so, we make use of the fact that the usual decoherence mechanism is, in a sense, simply a purified version of the system-environment interactions that are toy-modelled by noisy operations.
Let $\ket{\psi} \in \mc H_S$ be an initial state vector of a $d$-dimensional system and $\ket{\phi}$ be the initial state 
vector of the environment. 
According to the decoherence mechanism, the unitary joint evolution of system and environment is generated by a Hamiltonian whose interaction term picks out, or einselects, a preferred basis in which it decoheres the system~\cite{Zurek2003}. 
We are now interested in the smallest possible size of the environment that achieves this. 
Let us label the system basis that is einselected by $A = \{ \ket{i}\}$ and assume that $\ket{\phi}$ is a maximally entangled $d$-dimensional and bi-partite state vector over systems $E_1$ and $E_2$.
We then define the unitary
\begin{align} \label{eq:unitary_deco}
    U = U_{SE_1} \otimes {\mathbbm 1}_{E_2},
\end{align}
where $U_{SE_1}$ is the unitary defined in~\eqref{eq:dephasing_unitary_basic} that acts on systems $S$ and $E_1$. 
As is clear from the above, this unitary will have the effect that 
\begin{align}
    \tr_E[U \proj{\psi} \otimes \proj{\phi} U^\dagger] = \pi_A(\proj{\psi}), 
\end{align}
meaning that even in this purified picture only an environment of the size of the system is required to produce decoherence.

Let us now turn to the smallest possible measurement device.
For simplicity, we only consider projective measurement schemes: 
Suppose we are given a system in some initial state vector $\ket{\psi}$ and some set of projective measurement operators $\{M_i =
\proj{i}\}, i \in \{1, \dots, d\}$.
Then a measurement process consists of the following steps: 
A bi-partite measurement device, initially in state vector $\ket{\phi}$, consisting of a $d$-dimensional pointer system $P$ and a remainder $R$, whose dimension we are interested in bounding; 
and a unitary $W$ with the effect that 
\begin{align} \label{eq:measurement_scheme}
    Tr_R[W \proj{\psi} \otimes \proj{\phi} W^\dagger] = \sum_i p_i \proj{i, P_i}, 
\end{align}
where $p_i = \Tr(M_i \proj{\psi})$ and $\{\ket{P_i}\}$ form an orthonormal basis for the pointer system. Using the above results, we can
easily construct a measurement process as follows: Let the initial state vector of the
measurement device be $\ket{\phi} = \ket{0}_P \otimes \ket{\phi^+}_R$, where
$\ket{\phi^+}$ is a bi-partite, $d$-dimensional, maximally entangled state vector.
Further, let $\{V_i\}$ be unitaries defined by the action
\begin{align}
    V_i \ket{i, 0} = \ket{i,P_i}.
\end{align}
Finally, define the unitary 
\begin{align}
	W = \sum_i \proj{i}\otimes V_i  \otimes (U_i)_{R_1} \otimes \mathbbm{1}_{R_2},
\end{align}
where the unitaries $U_i$ form an operator basis as before. 
Then, it is easy to verify that $\ket{\phi}$ and $W$ together satisfy~\eqref{eq:measurement_scheme}. 
This shows that in principle one requires a measurement device (including the pointer variable) whose size is only twice that of the system to be measured to implement a projective measurement as a physical process. 
Using entropic arguments one can again show that this is also the smallest possible measurement device. 
Note that the register $R$ is exclusively used as a source of randomness in this protocol. 
Thus if we are willing to give up the assumption that the initial state of the measurement-device is pure, then it suffices to only keep part $R_1$ in a maximally mixed state. 
Clearly, these results can also be read as providing the minimal dimension of an environment
that equilibrates a quantum system of dimension $d$~\cite{LongReview,Linden2012}.

\subsection{A universal dephasing machine}
\label{sub:universaldecoherence}

In Section~\ref{sec:an_optimal_dephasing_map} we have shown that with the aid of a noise system $R$ in state $\id_{\lceil d^{1/2} \rceil}$ it is possible to perform a protocol $U$ which has the effect of implementing the dephasing map $\pi_A$ on the system $S$. 
We will now investigate which map is induced on $S$ if the same unitary is applied with a system $R$ in a state $\sigma$ different from $\id_{\lceil d^{1/2}\rceil}$. 
We will show that $U$ brings the system closer to $\pi_A(\rho)$ for any initial states $\rho$ and $\sigma$. 
Also, we find that iterating the same protocol $U$ with a sufficiently large sequence of imperfect noise states of $R$ brings the system $S$ exponentially close (in the number of iterations) to its dephased state. 
In this sense, $U$ acts as a universal dephasing machine (Fig.~\ref{fig:machine1} and Fig.~\ref{fig:machine}): 
an iterated use of the same protocol $U$ dephases the state of $S$ for large families of states on $R$ acting as a SoR. 
Hence one can implement this protocol universally as a ``black box'', without having to know the actual state of $R$.

\subsubsection{Imperfect noise and convergence to the dephased state}\label{sec:convergence}

Let $ \mathcal{D}_\sigma(\cdot)$ denote the map
\begin{align} \label{eq:dephasing_map}
\mathcal{D}_\sigma(\cdot) \coloneqq \Tr_R[U \: (\cdot \: \otimes \: \sigma )\: U^\dagger] 
\end{align}
where $U$ is the unitary of Lemma~\ref{lem:dephasing_basic}. In Appendix~\ref{app:universal_dephasing_machine} we show that, for any $\rho$ and $\sigma$, 
\begin{align}
    \mathcal{D}_\sigma(\pi(\rho)) = \pi(\mathcal{D}_\sigma(\rho)) = \pi(\rho),\\
    \norm{\mathcal{D}_\sigma(\rho) - \pi(\rho)}_1 \leq \norm{\sigma - \id_{\lceil d^{1/2} \rceil}}_1, 
\end{align}
where we have dropped the subscript $A$. 
These properties imply that, independently of the actual state $\sigma$, the system $S$ is brought closer to the dephased state $\pi(\rho)$ while keeping its diagonal invariant.
This follows from the data-processing inequality~\cite{Nielsen2000}
\begin{align}
	\norm{\mathcal{D}_\sigma(\rho) - \pi(\rho)}_1 = \norm{\mathcal{D}_\sigma(\rho) - \mathcal{D}_\sigma(\pi(\rho))}_1 \leq \norm{\rho - \pi(\rho)}_1.\nonumber
\end{align}

Using those properties, one can show that by repeating the process sequentially (see Fig.~\ref{fig:machine1} (top)) the system is eventually dephased for large classes of states $\sigma$. 
In fact, one can show that (see again Appendix~\ref{app:universal_dephasing_machine})
\begin{align}\label{eq:convergence_system}
  \norm{\mathcal{D}^n_\sigma(\rho) - \pi(\rho)}_1 \leq \norm{\sigma - \id_{\lceil d^{1/2} \rceil}}_1^n ,
\end{align}
where $\mathcal{D}^n_\sigma(\rho)$ denotes the repeated application of $\mathcal{D}_\sigma$. This means that, given $\sigma$ such that $\norm{\sigma - \id_{\lceil d^{1/2} \rceil}}_1<1$, the dephased state is approached exponentially fast.
Note that another corollary of the above properties is that the map $\mc{D}_{\sigma}$ can only increase the von Neumann entropy of its input, which is formally proven in Appendix \ref{app:robustness_noise}.

\begin{figure}[!t]
\includegraphics[width=0.4\textwidth]{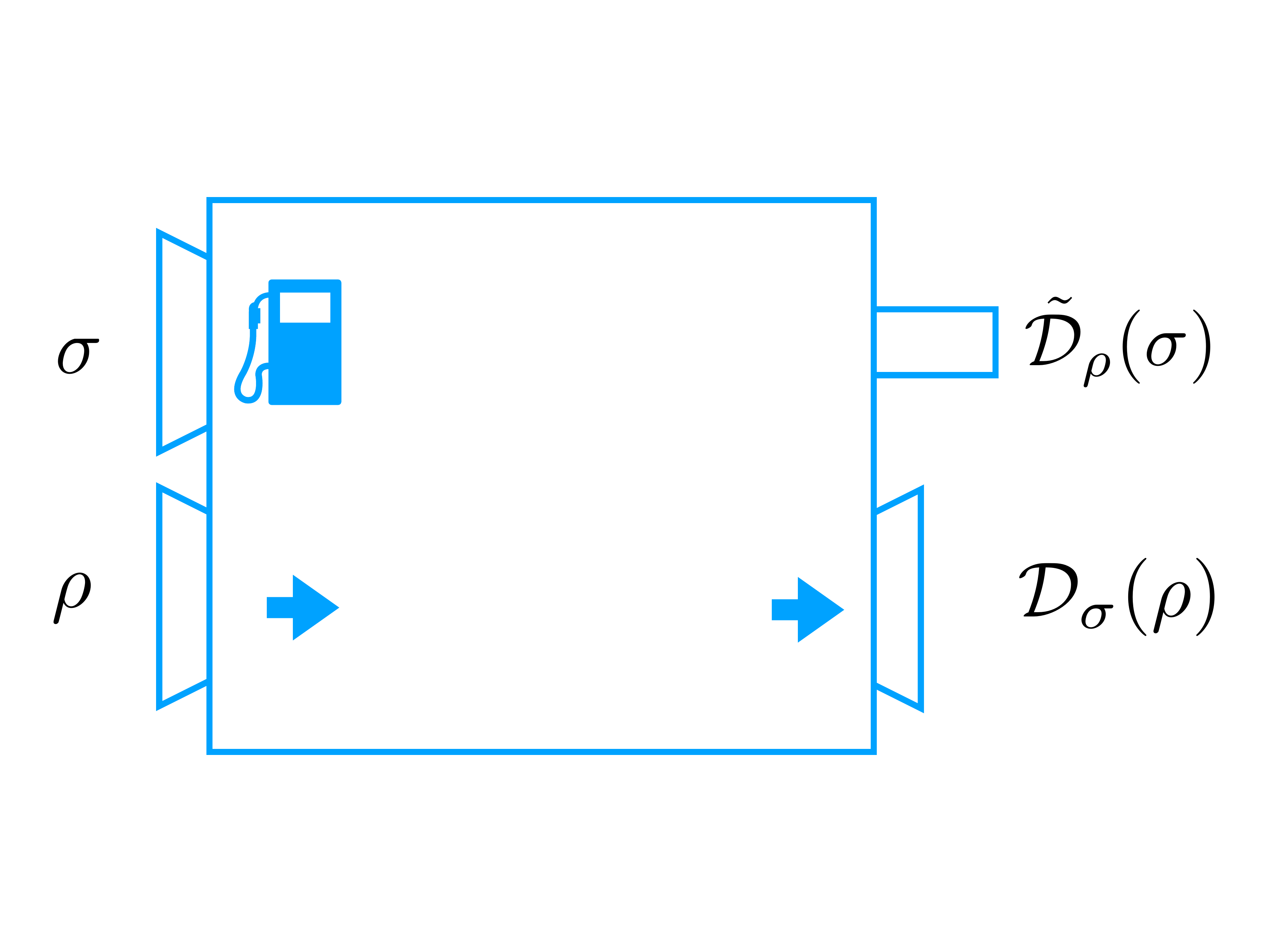}
\caption{Single instance of ``universal dephasing machine''. We interpret the process $\rho \otimes \sigma \to U(\rho \otimes \sigma)U^\dagger$ as a dephasing machine that takes the state $\sigma$ as fuel and transfers the input state $\rho$ into the output state $\mathcal{D}_\sigma(\rho)$ and ``waste'' $\tilde{\mathcal{D}}_\rho(\sigma)$.}
\label{fig:machine1}
\end{figure}

\subsubsection{Reusing the randomness} 
\label{ssub:reusing_the_randomness}

\begin{figure*}[!t]
\includegraphics[width=0.8\textwidth]{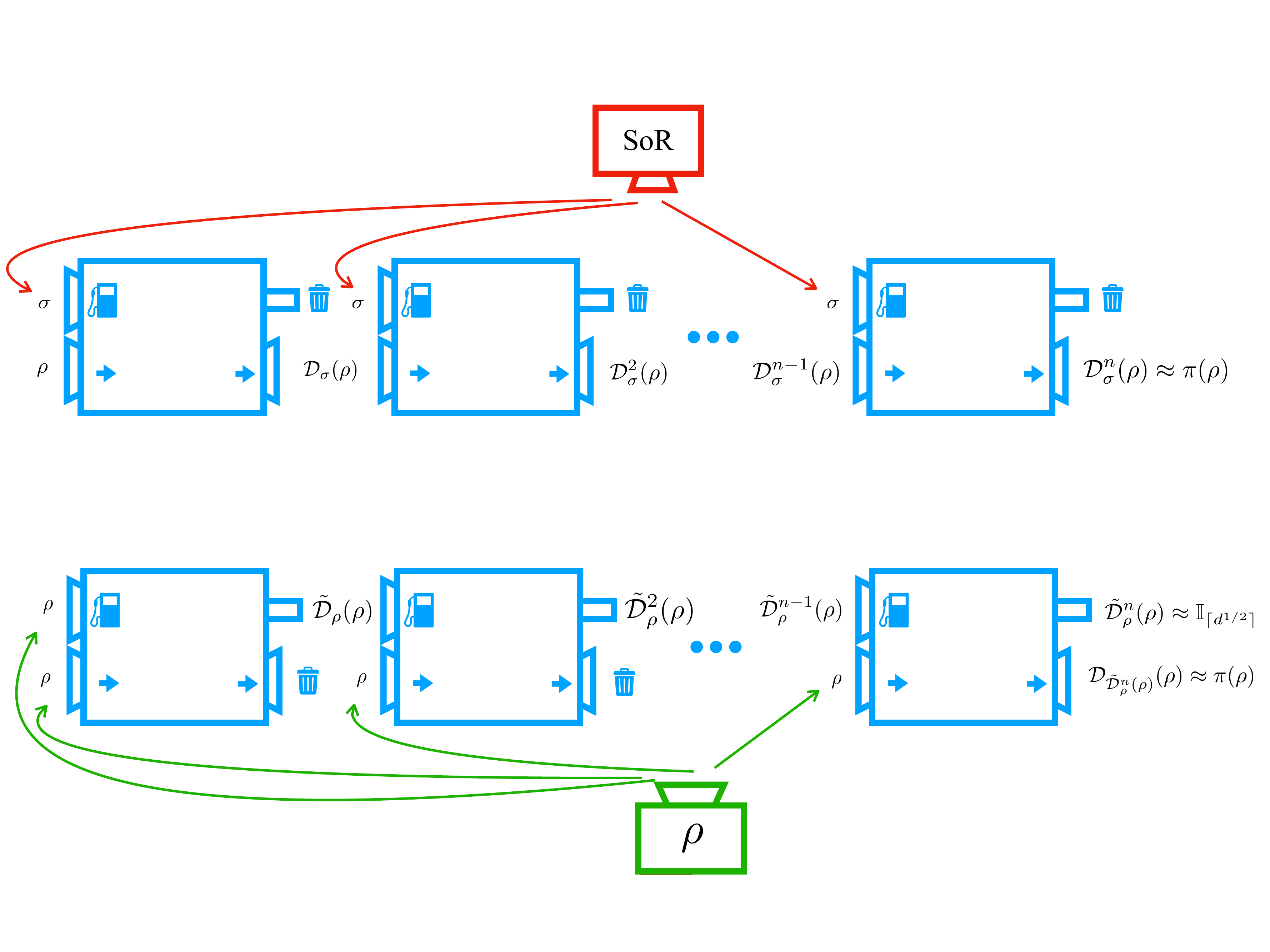}
\caption{Top: Repeated application on single input state approximates dephasing map. Bottom: Producing the dephased state when there is no SoR. If $\norm{\rho-\id_d}_1 < 1$, then the necessary amount of randomness for dephasing can be distilled by repeated application of the universal dephasing machine.}
\label{fig:machine}
\end{figure*}

In the case of $R$ being in the state $\id_{\lceil d^{1/2} \rceil}$, we have shown in Section~\ref{sub:catalyticity} that it remains unchanged and thus, the noise is re-usable. 
A natural question is then what happens to the state of $R$ when it is in an arbitrary state $\sigma$. 
Let $\tilde{\mathcal{D}}_\rho$ denote the map
\begin{align}
\tilde{\mathcal{D}}_\rho(\cdot) \coloneqq \Tr_R[U \: (\rho \: \otimes \: \cdot \:) U^\dagger].
\end{align}
It follows simply from Eq.~\eqref{eq:map_on_noise} that $\tilde{\mathcal{D}}_\rho$ is just a mixture of unitaries, hence bringing $R$ closer to the maximally mixed state. 
Indeed, following arguments analogous to the ones of Section~\ref{sec:convergence} (see Appendix~\ref{app:universal_dephasing_machine}) one can show that there exist choices for the unitary operator basis of Lemma~\ref{lem:dephasing_basic} so that the final state of $R$ fulfills 
\begin{align}
    \norm{\tilde{\mathcal{D}}_\rho(\sigma) - \id_{\lceil d^{1/2} \rceil}}_1 \leq \norm{\rho - \id_d},
\end{align}
and analogously it converges as
\begin{align}\label{eq:convergencenoise}
  \norm{\tilde{\mathcal{D}}^n_\rho(\sigma) - \id_{\lceil d^{1/2} \rceil}}_1 \leq \norm{\rho - \id_d}_1^n.
\end{align} 
Altogether we conclude not only that the noise can be re-used, but furthermore, that it improves its quality converging exponentially fast to a state of perfect noise, provided that the initial state $\rho$ is mixed enough to start with (as given by the condition $\norm{\rho - \id_d}_1<1$). 
The fact that the noise system is brought closer to the maximally mixed state allows one to implement a distillation protocol such as the one depicted in Fig.~\ref{fig:machine} (bottom). 
There, one has a single source providing copies of a given initial state $\rho$. 
One aims at dephasing each subsystem locally, similarly to what is done with a perfect noise system in Eq.~\eqref{eq:local_dephasing}. 
Here, one can take one copy $\rho$ playing the role of $R$ for some iterations until it is brought close enough to the maximally mixed state, which will happen exponentially quickly, given~\eqref{eq:convergencenoise}. 
Then, using Eq.~\eqref{eq:convergence_system} one can ensure that all the new copies of $\rho$ can be locally dephased.

\subsubsection{Time control for the dephasing machine and recurrence} 
\label{ssub:time_control_for_the_dephasing_machine}

So far we have left unspecified how the dephasing of the machine would physically be implemented. 
One concern here may be that the dephasing properties heavily rely on very precise time control of the evolution under the associated Hamiltonian $H = \mathrm{i} \log(U)$. 
However, the numerical simulations depicted in Fig.~\ref{fig:robust_simulation} strongly indicate that, as the system dimension becomes large, $H$ produces an evolution that is close to $\mc D_\sigma(\cdot)$ for a time-span that scales exponentially with the size of $S$. Indeed, for prime power dimensions and the case $\sigma = \id_{\lceil d^{1/2} \rceil}$, we find analytically that integer iterations of the application of the dephasing unitary always yield the exact dephasing map, up to a recurrence point, at which the original state is returned. See Appendix~\ref{app:timing} for details. The numerical simulations above complement this and suggest that this recurrence property holds not only for integer iterations of the application of the dephasing unitary, but also for intermediate times. 

We hence expect that in the limit of very large dimensions, this \emph{equilibrating} behavior~\cite{LongReview,Linden2012}
becomes arbitrarily good and the state $\rho(t)$ remains close to the equilibrium state $\pi(\rho)$ for a time exponential in the system size. This means that the universal dephasing machine can be made robust in time, in the
 sense that it does not require exact control over the timing and the dephasing
 is maintained for long time scales.  

 \begin{figure}[!tbp]
 	\includegraphics[width=0.4\textwidth]{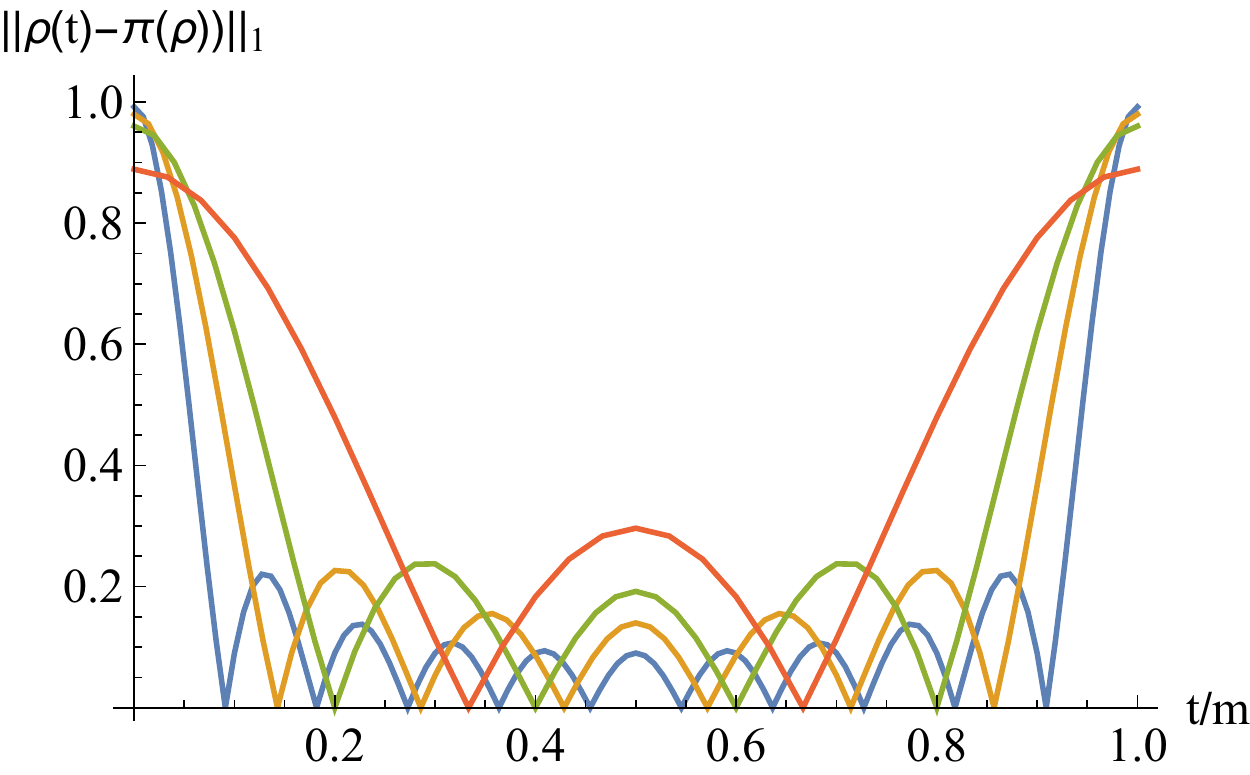}
 	\caption{Numerical simulations of the dephasing map that is induced by
 	the noisy operation~\eqref{eq:dephasing_unitary_basic} for continuous
	 time and system dimensions $d=m^2=9,25,49,121$ (red,green,yellow,blue).
 Shown is the trace-norm distance between the time evolved state $\rho(t)$ and the
	 pinched state $\pi(\rho)$ as a function of rescaled time $t/m$. The initial state is a maximally coherent state $\frac{1}{\sqrt{d}}\sum_i \ket{i}$. The graph shows that, while for integer times (with
 respect to the dimension of the environment) the dephasing is always exact, for non-integer times
	 the deviation from exact dephasing becomes small with increasing dimension. The numerically obtained deviation at $t/m=0.5$ seems compatible with a scaling as $1/m=1/\sqrt{d}$, but we leave open to derive the exact scaling behaviour.}
 	\label{fig:robust_simulation}
\end{figure}

\subsection{An entanglement-assisted private quantum channel} 
\label{sec:a_modified_private_quantum_channel}

In this section, we apply our results to the construction of a cryptographic protocol known as private quantum channel (PQC). 
In a PQC-setting, two parties, Alice and Bob, would like to communicate quantum data privately, that is, without an eavesdropper being able to intercept and retrieve the data. 
To achieve this they share a secret key.
We will now first briefly explain PQCs using classical secret keys and then provide a construction where the classical key $k$ is substituted for a "quantum key" in the form of a minimal number of entangled bits. 
In the following, we denote by $\mc S(\mc H)$ the set of normalized quantum states on the Hilbert space $\mc H$. 
Formally, in the classical-key setting, a $(\delta, \epsilon)$-PQC is a set of pairs of encoding and decoding CPTP-maps $\mc X_k: \mc S(\mc H_A) \to \mc S(\mc H_{A'})$ and $\mc Y_k: \mc S(\mc H_{A'}) \to \mc S(\mc H_A)$ that can be locally implemented by the sending and receiving parties respectively, where $k$ denotes the secret key that is shared by Alice and Bob. 
We think of the key $k$ as a random variable and assume that the key $k$ occurs with probability $p(k)$. 
These channels then have to fulfill the following conditions \cite{Ambainis2009}. 
Firstly, there exists a fixed element $\tau \in \mc S(\mc H_A')$, such that 
\begin{align} \label{eq:pqc_reliability}
    \sup_{\rho_{A,B} \: \in \: \mc S(\mc H_A \otimes \mc H_B)} \norm{ \left(\sum_k p_k\mc X_k \otimes \text{id}\right)(\rho_{A,B}) - \tau \otimes \rho_B}_1 \leq \epsilon,
\end{align}
where $\rho_{A,B}$ is any extension of the input state $\rho_A$ to a larger Hilbert space and $\rho_B = \tr_A(\rho_{A,B})$. And secondly,
\begin{align} \label{eq:pqc_security}
    \sup_{\rho \: \in \: \mc S(\mc H_A)} \norm{\sum_k p_k \mc Y_k \circ \mc X_k (\rho) - \rho}_1 \leq \delta.
\end{align}
Eq.~\eqref{eq:pqc_reliability} warrants (approximate) security from eavesdropping, while~\eqref{eq:pqc_security} warrants the channel's (approximate) reliability. 
The reason that the security is defined over all possible extensions is that the eavesdropper may initially be entangled with part of the unencrypted message. 
Finally, a $(0,0)$-PQC is called an \emph{ideal} PQC.

PQCs have been well-studied for the case in which Alice and Bob share a classical key~\cite{Boykin2003Optimal,Ambainis2000,Hayden2004Randomizinga,Ambainis2009,Portmann2017,Hayden2016Universal}. 
In this case{, and if $\mc X_k$ is unitary,} the encoding corresponds to a classical noisy process and a key of length at least $(2-O(\epsilon))n$ is necessary for the $\epsilon$-secure transmission of $n$ qubits~\cite{Ambainis2009,Boykin2003Optimal,Ambainis2000,Note1}.

Here, in contrast, we consider a setting in which Alice and Bob share a ``quantum key'' in the form of entangled quantum states.  
We use our dephasing map to construct an ideal private quantum channel that requires $n$ shared ebits of entanglement to transmit $n$ qubits of quantum data. 
As with the dephasing map, this value can again be shown to be optimal, in the sense that no implementation of an ideal PQC as a noisy operation can require fewer ebits (a result that extends to approximately ideal PQCs).
It improves on the only other discussion of PQCs that uses entanglement known to the authors, in Ref.\ \cite{Leung2002Quantum}. 
There, an ideal PQC is constructed that applies techniques from classical PQCs and hence achieves only ``classical'' efficiency by requiring $2n$ ebits for $n$ transmitted qubits.

\begin{figure}[!t]
\includegraphics[width=0.28\textwidth]{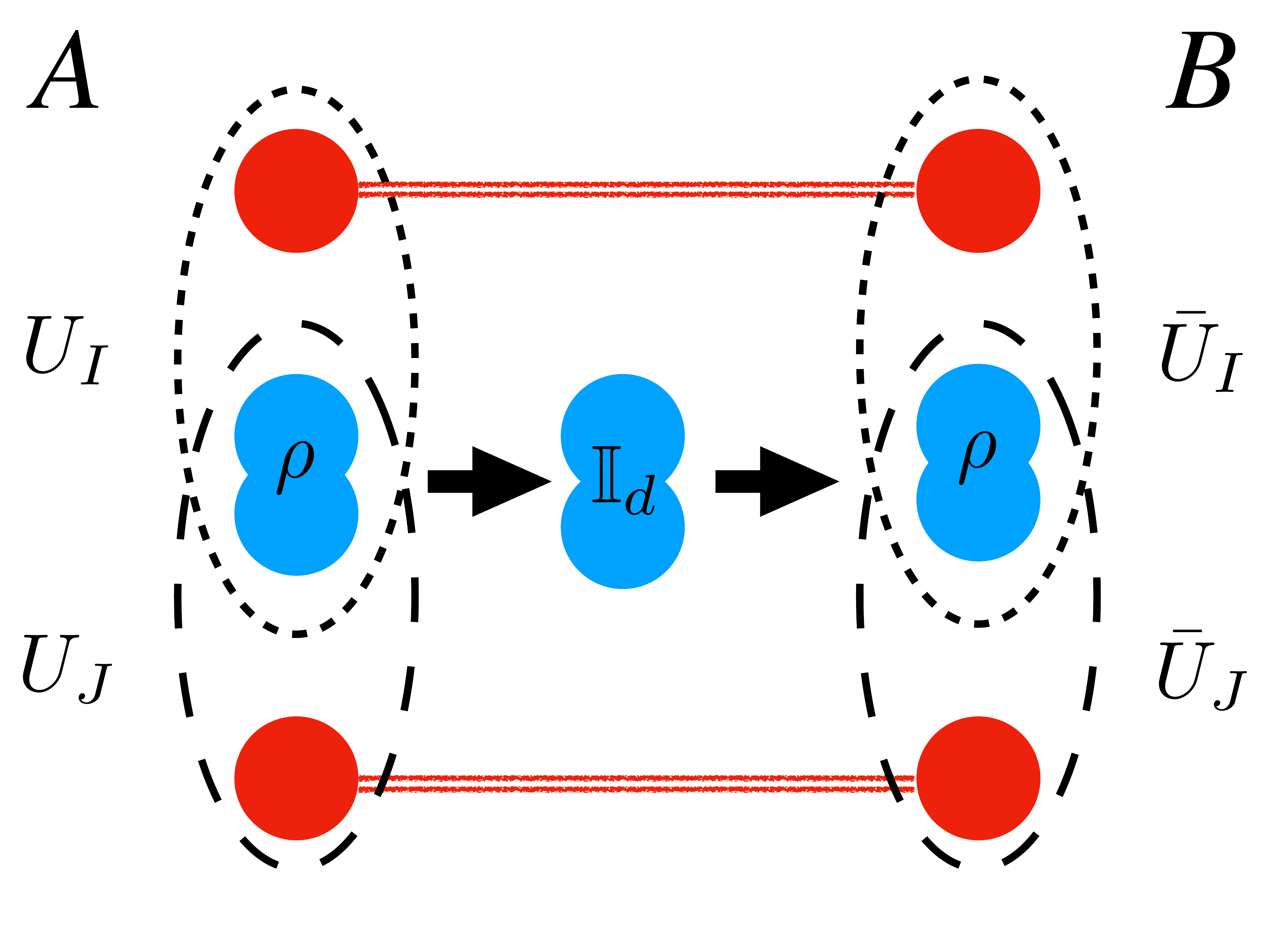}
\caption{Illustration of our quantum PQC for the case $n=2$. To encode a two-qubit state $\rho$ (blue), Alice applies the dephasing unitaries $U_I$ and $U_J$ to the system and one half of an ebit (red) each, where $I$ and $J$ can be any mutually unbiased bases. This maps $\rho$ into the maximally mixed state exactly, so that an eavesdropper cannot learn anything about $\rho$ even if she was initially entangled with part of it. Bob, in order to decode, applies the conjugate of the two above unitaries and thereby retrieves the state exactly.}
\label{fig:pqc}
\end{figure}

The idea behind our construction is straightforward (see Fig.~\ref{fig:pqc}). Given an $n$-qubit system $S$, let $U_I$ and $U_J$ denote the dephasing unitaries~\eqref{eq:dephasing_unitary_basic} whose projective part corresponds to the two orthonormal bases $I = \{\ket{i}\}_{i=1}^d$ and $J = \{\ket{j}\}_{j=1}^d$ for $\mc H_S$. If Alice and Bob share $n$ ebits, and assuming for convenience that $n$ is even, Alice can split the ebits into two halves, which we call $E_1$ and $E_2$. She then applies $U_I$ to $S$ and her local share of $E_1$, followed by applying $U_J$ to $S$ and her half of $E_2$. It is easy to check that if $I$ and $J$ are mutually unbiased, that is, if 
\begin{align}
    |\braket{i}{j}|^2 = \frac{1}{d}, \quad \forall i,j,
\end{align}
then this results in the completely depolarizing channel. That is, the map
\begin{align}
    \mc X(\cdot) \coloneqq \tr_{E}(U_J U_I( \cdot \otimes \:\proj{\phi^+}_{E_1} \otimes \proj{\phi^+}_{E_2}) U^\dagger_I U^\dagger_J),
\end{align}
where $\ket{\phi^+}$ represents an $n/2$-ebit state vector, has the property that 
\begin{align}
    \mc X(\rho) = \id_d, \quad \forall \rho \in \mc D(\mc H_S).
\end{align}
This ensures perfect secrecy, since the completely depolarizing channel necessarily also removes all correlations to other systems \cite{Hayden2004Randomizinga}.
Upon receipt of $S$, Bob can then apply the complex conjugate of the encoding unitaries to his share of the ebits to retrieve the original state. See Appendix~\ref{app:pqc} for the formal proofs. 

This construction has a number of interesting features, some of which, however, are already present in the construction of Ref.\ \cite{Leung2002Quantum}. For instance, it is catalytic in the sense that, at the end of the transmission process, in case no eavesdropper has interacted with the sent data, all of the entanglement is returned in its initial state and can be reused for future rounds of transmission. Moreover, the scheme allows for error correction, efficient authentication and recycling of some of the entanglement in case eavesdropping has occurred. We refer the reader to Appendix~\ref{app:pqc} for a discussion of these properties.

\section{Dephasing with quantum expanders}\label{sec:expandergraphs}

The protocol presented in Lemma~\ref{lem:dephasing_basic} allows one to dephase perfectly a $d$-dimensional system
given a SoR of dimension of $m=\lceil d^{1/2}\rceil$. This very same protocol, when applied to an imperfect SoR of dimension $m$ but not in the maximally mixed state, yields, as shown in Section~\ref{sec:convergence}, a convergence to the dephased state when the protocol is iterated. 
In this section we study a complementary protocol that provides fast convergence when we have states of the SoR that are maximally mixed, but of dimension significantly smaller than $m$. We find a protocol that yields an exponential convergence to the dephased state with the system size (i.e., the logarithm of Hilbert space dimension) of the SoR, measured in the $2$-norm. This is remarkable, in that it shows that one can obtain an equilibration in $2$-norm exponentially quickly in the ancillary system size. This insight may be seen as being at odds with the intuition that an equilibrating environment should naturally have a large physical dimension.
Our approach is based on a machinery of \emph{quantum expanders} 
\cite{PhysRevA.76.032315,Margulis,HarrowExpander}. The key insight is that one can trade 
residual correlations still present in the system with the dimension required for the mixing environment.

\begin{theorem}[Dephasing with quantum expanders]\label{thm:expanders} For any 
$d$-dimensional state, $d=e^2$ with $d$ odd, and an integer $k$, 
	there exists an $8^k$-dimensional
	quantum system $R$ and a unitary $U\in \mathrm{U}(d 8^k)$ such that
\begin{equation}
	\|{\rm tr}_R \left(U(\rho \otimes \id_{8^k}) U^\dagger\right) - \pi(\rho)\|_2 \leq  
	\sqrt{2 {d^3}} \left(
	 \frac{ 5 \sqrt{2}}{8}\right)^k.
\end{equation}
\end{theorem}

The restriction to the dimension is done for pure conceptual simplicity. The argument for the proof, presented in Appendix~\ref{sec:expanders_app}, follows from a construction of a classical 
random walk that acts on the vertices of an expander graph, a \emph{Margulis expander}
\cite{MargulisExpander}. In the present construction, 
the vertices of the Margulis expander 
are seen as lines labeled by $q=1,\dots, d$ in a $d\times d$-dimensional
quantum phase space of the $d$-dimensional quantum system. The central insight is that 
classical random walks on such lattices are reflected by random walks on Wigner functions
defined on $d\times d$-dimensional phase spaces,
which in turn give rise to random unitary channels on quantum states in $d$ dimensions. The
construction laid out in detail in the appendix builds upon and draws inspiration from
the scheme of Ref.~\cite{Margulis}, but is in several important ways a new scheme, in particular in that 
each line in phase space is treated separately. In this way,
the strong mixing properties of the random walk of the Margulis expander graph is not used to
show rapid mixing to a maximally mixed state, but in fact to a quantum 
state with vanishing off-diagonal elements. 

\section{Summary and conclusions}

We have studied the problem of implementing state transitions under noisy processes, that is, processes that require randomness. We solve this problem completely by providing optimal protocols for both the case of an implicit, classical model of randomness as well as an explicit, quantum model of randomness. The main building block behind these protocols is the construction of a protocol that performs a dephasing map on an artibrary quantum state using a SoR of the smallest possible dimension, both for the quantum and classical case. We find that a quantum SoR is quadratically more efficient than its classical counterpart due to quantum correlations, and hence show that an explicit model is strictly more powerful for any dimension $d > 2$.

{Once the optimal protocols for dephasing were established, we studied applications such as state transitions in noisy operations, decoherence and quantum measurements, providing optimal protocols for all of them. An interesting feature of our protocol is that the SoR is not altered during the protocol, meaning that it can be re-used to implement further iterations of the above tasks.} 

{We have also extended our discussion to the case of imperfect noise and used our results to construct a universal dephasing machine that exhibits robustness both with respect to the noise that fuels it, as well as with respect to the control over timing when running it. Moreover, we have used our dephasing as a primitive to construct a novel, ideal private quantum channel. Finally, by putting it into the context of expander graphs, we have seen how 
such a dephasing is possible very economically, leading to an approximate dephasing in $2$-norm.}

{Besides the foundational interest of our construction, which makes precise the way in which the relationship between correlations and randomness in quantum mechanics differs from that in classical mechanics, we expect our dephasing protocol to improve bounds in noisy processes that we have not discussed here, to the extent that 
introduce a new primitive to constructions in quantum information. Given 
the pivotal status of randomness in protocols of quantum information processing and in
notions of quantum thermodynamics, these results promise a significant number of further
practical  
applications.} 

\emph{Acknowledgements:} P.~B.~thanks Lluis Masanes, Markus M\"{u}ller, Jon Richens and Ingo Roth for interesting conversations and especially Jonathan Oppenheim for suggesting cryptographic applications of the results. We acknowledge funding from the ERC (TAQ), the DFG 
(EI 519/14-1, CRC183), the Templeton Foundation, and the Studienstiftung des Deutschen Volkes. 
\bibliographystyle{unsrtnat}

\appendix

\section{Lower bounds on dimension of source of randomness}
\label{sec:necessary}
In this section we prove the lower bounds in Lemma~\ref{lemma:optimaldimensions}. In fact, we prove them in an approximate setting to show that they are robust to small deviations from exact dephasing. 
To do so, call a map $\mc E^m_X$ \emph{$\epsilon$-dephasing} if, for all operators $\rho \in \mc B(\mc H_S)$ and some fixed basis $A$,
	 	 \begin{align} \label{eq:norm_bound}
	     \norm{ \mc E^m_X(\rho) - \pi_A(\rho)}_1 \leq \epsilon, 
	 	\end{align}
where $X \in \{C,Q\}$. 
Let $m^*_X(d,\epsilon)$ be the smallest value of $m$ such that an $\epsilon$-dephasing map can be realised as a map of the form \eqref{eq:def:classicalnoise} for $X=C$ and \eqref{eq:def:quantumnoise} for $X=Q$ respectively, $\dim(\mc H_S) = d$.

We begin with the classical bound. Consider the state vector
\begin{align}
	\ket{A} \coloneqq \frac{1}{\sqrt{d}} \sum_i \ket{i}. 
\end{align}
If it is dephased in the basis $A=\{\ket{i}\}$, it is mapped to the maximally mixed state.
We are concerned with deriving the minimal value of $m$ such that $\mc{E}_C^m(\ketbra{A}{A})=\id_d$.
For this, note that
\begin{align}
	{\mc{E}_C^m(\ketbra{A}{A})} = \frac{1}{m}\sum_{j=1}^m U_j \proj{A}U_j^\dagger.
\end{align}
Clearly, this state has at most rank $m$, since its support is spanned by $m$ vectors. Moreover, it is easy to see that for any $\epsilon$-dephasing classical map $\mc{E}_C^m$, 
\begin{align} \label{eq:rank_bound}
    \text{rank }\mc{E}_C^m(\proj{A}) \geq d(1-\frac{\epsilon}{2}),
\end{align}
which implies
\begin{align}
	m^*_C(d, \epsilon) \geq m \geq \max\left\{2,d(1-\frac{\epsilon}{2})\right\},
\end{align}
where we also used that any non-trivial source of randomness must be at least two-dimensional.

To see \eqref{eq:rank_bound}, consider any state $\rho$ of rank $k$. Then 
\begin{align}
    \norm{\rho - \id_d}_1 \geq \norm{\id_{k,d} - \id_d}_1 = 2(1-\frac{k}{d}),
\end{align}
where $\id_{k,d}$ is a $d$-dimensional state that is maximally mixed on a subspace of dimension $k$ (and hence has rank $k$). 
Using \eqref{eq:norm_bound} and re-arranging then gives bound $\eqref{eq:rank_bound}$.

Let us now turn to the quantum case, where we find
     \begin{align} \label{eq:quantum_bound}
	     m^*_Q(d, \epsilon) \geq \max\left\{2,d^\frac{1-\epsilon}{2} \epsilon^{\frac{\epsilon}{2}}\right\},\quad \forall \epsilon\leq \frac{1}{6e}.
     \end{align}
     First note that for $d\leq 4$, our optimal construction already yields $m=2=\lceil d^{1/2}\rceil$ and that any non-trivial source of randomness must have $m\geq 2$. In the following, we hence assume $d\geq 5$. 
Now consider again the initial state $\proj{A}$. Then, for any $\epsilon$-dephasing map $\mc E_Q^m$, applying Fannes' inequality yields
\begin{align} \label{eq:ent_bound}
	S(\mc E^m_Q(\proj{A})) \geq \log d + \epsilon \log \left(\frac{\epsilon}{d}\right). 
 \end{align}
 In the following, let $\rho_R'$ denote the state on the $m$-dimensional source of randomness after the dephasing map has been applied. 
 From our construction of the exact dephasing map, we know that $m^*(d,\epsilon)\leq \lceil  d^{1/2}\rceil$. 
 Hence, in the following we assume $2\leq m\leq \lceil d^{1/2}\rceil$. Since $\epsilon \leq 1/6e$ and 
 \begin{align}
	 \log\left(\left\lceil d^{1/2} \right\rceil\right) - \log\left( d^{1/2}\right) \leq 1/2,\quad \forall d\geq 5,	 
\end{align}
it follows using \eqref{eq:ent_bound} that
\begin{align}
	S(\mc E^m_Q(\proj{A})) > \log\left(\left\lceil d^{1/2}\right\rceil\right) \geq S(\rho_R').
\end{align}
We finally use the Lieb-Araki triangle inequality, which states that
\begin{align}
	S(\rho_{A,B}) \geq \left|S(\rho_A) - S(\rho_B)\right|
\end{align}
for any bipartite state $\rho_{A,B}$. We can now use this to bound
\begin{align}
	\log m &= S(\proj{A}) + S(\id_m) = S(U\proj{A}\otimes\id_m U^\dagger) \\
	&\geq \left|S(\mc E^m_Q(\proj{A})) - S(\rho_R') \right| \\
	&= S(\mc E^m_Q(\proj{A})) - S(\rho_R')\\
	&\geq \log(d) + \epsilon \log(\epsilon/d) - \log m. 
\end{align}
Hence, we obtain
\begin{align}
	m \geq d^{\frac{1-\epsilon}{2}} \epsilon^\frac{\epsilon}{2},
\end{align}
which finishes the proof. 
%

\section{The universal dephasing machine}
\label{app:universal_dephasing_machine}
In this appendix, we provide further details on the results regarding the universal dephasing machine. For convenience, we drop the subscripts for the dephasing maps and the maximally mixed states.

\subsection{Robustness with respect to imperfect noise}
\label{app:robustness_noise}
Let us first show the following lemma.
\begin{lemma}[General properties of $\mc D_\sigma$]\label{lemma:properties_D}
The family of channels $\mc D_\sigma$ has the following properties:
\begin{enumerate}
\item \emph{(Fixed points)} All diagonal states are fixed points:
\begin{align}\mc D_\sigma(\pi(\rho))=\pi(\rho),\quad \forall \sigma,\rho. \end{align}
\item \emph{(Invariant diagonal)} The channels do not modifiy the diagonal of any state in the given basis:
\begin{align}
	\pi(\mc D_\sigma(\rho)) = \pi(\rho), \quad \forall \sigma,\rho.
\end{align}
\item \emph{(Continuity)} 
 The following continuity property holds: 
\begin{align}
 \norm{\mc D_\sigma(\rho) - \pi(\rho)}_1 &\leq \norm{\sigma - \id}_1.
\end{align}
\end{enumerate}
\end{lemma}

\begin{proof}
The first two properties follow from the definition of $\mc D_\sigma$ in \eqref{eq:dephasing_map}, since
\begin{align}
    \bra{k}\Tr_B(U(\rho \otimes \sigma)U^\dagger)\ket{k} &= \sum_{i,j} \braket{k}{i}\bra{i}\rho \ket{j}\braket{j}{k} \Tr(U_i \sigma U_j^\dagger) \\
    &=  \rho_{k,k} \Tr(U_k^\dagger U_k \sigma)=\rho_{k,k}. 
\end{align}
The continuity property can be seen as
\begin{align}
 \norm{\mc D_\sigma(\rho) - \pi(\rho)}_1 &= \norm{\tr_B(U(\rho \otimes \sigma)U^\dagger) - \tr_B(U(\rho \otimes \id)U^\dagger) }_1 \nonumber\\
 	&\leq \norm{U(\rho \otimes (\sigma - \id))U^\dagger}_1 \nonumber\\
 	&= \norm{\rho \otimes (\sigma - \id)}_1 \nonumber\\
 	&= \norm{\sigma - \id}_1,
\end{align}
where we have used the data-processing inequality and the unitary invariance of the norm. 
\end{proof}

In particular, the fixed-point property has the following corollaries.
\begin{corollary}[Contraction to dephased state]
Let $f(\rho,\rho')$ be any measure of distance between quantum states that fulfills the data-processing inequality, for example, any Renyi-divergence or the trace-distance~\cite{Nielsen2000}. 
Then
\begin{align}
f(\rho,\pi(\rho)) \geq f(\mc D_\sigma(\rho),\pi(\rho)),\quad \forall \sigma.
\end{align}
\end{corollary}
Choosing $f(\rho,\sigma)$ as the quantum relative entropy $S(\rho\|\sigma)$ and using that $S(\rho \| \pi(\rho)) = S(\pi(\rho)) - S(\rho)$ we then obtain the following corollary.
\begin{corollary}[Increasing entropy] The channels $\mc D_\sigma$ can only increase the von~Neumann entropy:
\begin{align}
S(\rho) \leq S(\mc D_\sigma(\rho)),\quad \forall \sigma.
\end{align}\label{cor:increasing_entropy}
\end{corollary}

So far we have only considered single applications of the dephasing map. 
Let us now consider repeated applications. 
We thus want to investigate what happens if we have a stream of sources of randomness $\sigma_i$ and sequentially use them to dephase the system. 
To this end, we can prove the following lemma:
\begin{lemma}[Iterated dephasing]
	Let $\{\sigma_i\}_{i=1}^n$ be arbitrary quantum states of dimension $\lceil d^{1/2}\rceil$. Then we have
	\begin{align}
		\norm{(\mc D_{\sigma_n}\circ \cdots \circ \mc D_{\sigma_1})(\rho) - \pi(\rho)}_1 \leq \Pi_{i=1}^n \norm{\sigma_i - \id}_1.	
	\end{align}
	\begin{proof}
		We prove the case $n=2$. The general result follows by iteration. First we use $\pi(\rho) = \pi\circ \mc D_\sigma(\rho)= \mc D_\sigma\circ \pi(\rho)$ to write
	\begin{align}
		\norm{(\mc D_{\sigma_2} \circ \mc D_{\sigma_1})(\rho) - \pi(\rho)}_1 &= \norm{(\mc D_{\sigma_2}-\pi) \circ (\mc D_{\sigma_1}-\pi)(\rho)}_1.\nonumber
  	\end{align}
	We can then estimate this norm as 
	\begin{align}
		\norm{(\mc D_{\sigma_2} \circ \mc D_{\sigma_1})(\rho) - \pi(\rho)}_1 &\leq \norm{\mc D_{\sigma_1}-\pi}_{1\to 1} \norm{\mc D_{\sigma_2}-\pi}_{1\to 1}, 
	\end{align}
	where $\norm{\cdot}_{1\to 1}$ is the norm on super-operators induced by the $1$-norm. From Lemma~\ref{lemma:properties_D} we can estimate it as
	\begin{align}
	\norm{\mc D_\sigma -\pi}_{1\to 1} = \max_\rho \norm{\mc D_\sigma(\rho)-\pi(\rho)}_1 \leq \norm{\sigma-\id}.
	\end{align}
	This step completes the proof. 
	\end{proof}
\end{lemma}
We thus find that $\rho$ converges exponentially quickly to the dephased state upon iterated application of $\mc D_\sigma$ provided that $\norm{\sigma_i-\id}_1\leq k<1$ for some $k$ and all $\sigma_i$. 

\subsection{Action on source of randomness}
\label{app:action_sor}
Let us now consider the action of the dephasing unitary on the source of randomness. Given some $\rho$, we are thus interested in the channel 
\begin{align}
	\tilde{\mc D}_\rho(\sigma) = \tr_S\left(U (\rho \otimes \sigma) U^\dagger\right).
\end{align}
This channel is always unital, i.e. it fulfills $\tilde{\mc D}_\rho(\id)=\id$ for any $\rho$.
Thus 
\begin{align}
	\norm{\tilde{\mc D}_\rho(\sigma) - \id}_1 \leq \norm{\sigma-\id}_1.
\end{align}
Let us denote by $\mathcal{R}$ the channel that maps any state into the maximally mixed state, $\mathcal{R}(\sigma)=\id$.
Then we have $\mathcal{R}=\tilde{\mc D}_\rho \circ \mathcal{R} = \mathcal{R}\circ \tilde{\mc D}_\rho$. By the same arguments as in the last section we then obtain the following lemma.

\begin{lemma}[Iterated mixing]
	Let $\{\rho_i\}_{i=1}^n$ be arbitrary quantum states of dimension $d$. Then we have
	\begin{align}
		\norm{(\tilde{\mc D}_{\rho_n}\circ \cdots \circ \tilde{\mc D}_{\rho_1})(\sigma) - \id}_1 \leq \Pi_{i=1}^n \norm{\rho_i - \id}_1.	
	\end{align}
\end{lemma}

\section{Recurrence and robustness in time}
 \label{app:timing}
 In this section we show that one can choose the operator basis $\{U_i\}$ from lemma~\ref{lem:dephasing_basic} in such a way that the dephasing map exhibits recurrence properties. 
 By recurrence we here mean that applying the dephasing unitary a certain number  of times undoes the dephasing, while it keeps it dephased for intermediate times. 
 
 To this end, note that one particular realization of this operator basis is the following: 
 Define the unitaries 
 \begin{align} \label{eq:def_special_uni}
      U_{r,s} \coloneqq \tau^{rs}X^r Z^s,
  \end{align}
  where $X, Z$ are the generalized Pauli matrices defined in~\eqref{X} and~\eqref{Z} respectively, and $\tau = -e^{\pi i/m} = -\sqrt{\omega}$. In the following, expressions are to be taken modulo $m$, unless specified otherwise. The conjugation relation $XZ=\omega^{-1}ZX$ then gives rise to the following properties in any dimension~\cite{Bengtsson_2017},
 \begin{align}
 	U_{r,s}U_{u,v} &= \omega^{us-vr}U_{u,v}U_{r,s} = \tau^{us-vr}U_{r+u,s+v}, \\
 	U_{r,s}^k &= U_{kr,ks} ,\\
 	U_{r,s}^\dagger &= U_{-r,-s},\\
 	\Tr(U_{r,s}) &=m \delta_{r,0}\delta_{s,0}.
 \end{align}
 These imply in particular that $\{U_{r,s}\}, r,s \in \{0, \dots, m-1\}$ form a unitary operator basis of $\mc B(\mc H)$. 
 Now, while it is clear that $X^m=Z^m=\id$, we can ask for the smallest $k$ such that $U_{r,s}^k=\id$ for all $r,s$. 
 The above conjugation relations imply that if $m$ is odd then this value is given by $m$, while for even $m$, the answer is $2m$. 
 For instance, in the case of $m=2$, we have $X^2=Z^2=\id$, while $(XZ)^2=-\id$. 
 Moreover, we can ask for the dependence of the \emph{order} of the unitaries $U_i$, by which we here mean the \emph{smallest} $k$ such that $U_i^k=\id$, i.e. the order of the corresponding element in the Weyl-Heisenberg group, on $m$. 
 Here, one has that the order of all non-trivial $U_i$ is $d$, if and only if $d$ is an odd prime. This special property for odd primes will be of key importance to establish recurrence relations in the following. 
 Define the map
 \begin{align}
     \pi^k_m(\cdot) &= 
 \begin{cases}
 	     \text{id}(\cdot), &\text{ if } k \text{ mod } m = 0, \\
 	     	\pi_A(\cdot), &\text{ otherwise, } 	
 	     \end{cases}
 \end{align}
where $A$ denotes the orthonormal basis in which the pinching acts, as in the main text. We then have the the following lemma.

\begin{lemma}[Recurrence for odd prime dimension]
Let $\dim{\mc H_S} = m^2, \dim{\mc H_R} = m$, where $m$ is an odd prime. There exists a unitary $V$ acting on $\mc H_S \otimes \mc H_R$ such that
\begin{align} \label{eq:prime_dilation}
     \Tr_B(V^k (\rho \otimes \id_m )(V^\dagger)^k) = \pi^k_m(\rho).
\end{align}
\end{lemma}
\begin{proof}
Let $A = \{\ket{r,s}\}_{r,s=1}^{m}$ be the orthonormal basis of $\mc H_S$ in which we want to pinch the state $\rho$. 
Define 
 \begin{align} \label{new_unitary}
     V = \sum_{r,s} \proj{r,s}_S \otimes (U_{r,s})_R,
 \end{align}
where the basis with respect to which the operators~\eqref{eq:def_special_uni} are defined can be chosen arbitrarily.  
Then, from the properties of these operators, we have 
\begin{widetext}
\begin{align}
     \Tr_R (V^k (\rho \otimes \id/d )(V^\dagger)^k) &= \sum_{r,s,u,v} \proj{r,s}\rho \proj{u,v} \frac{1}{m}\Tr(U_{kr,ks} U_{-ku,-kv})\\
     &= \sum_{r,s,u,v} \proj{r,s}\rho \proj{u,v} \frac{1}{m}\tau^{k^2(us-rv)}\Tr(U^k_{r-u,s-v}) \\
      &= \sum_{r,s,u,v} \proj{r,s}\rho \proj{u,v}\theta_m(k,r,u,s,v)\\
     &= \pi^k_m(\rho),
\end{align}
where the last line follows because  
\begin{align} \label{eq:phase_cancel}
    \theta_m(k,r,u,s,v) &\coloneqq \frac{1}{m}\tau^{k^2(us-rv)}\Tr(U^k_{r-u,s-v}) &=
         \begin{cases}
         	1, &\text{ if $k$ mod $m = 0$ or both $r=u$ and $s=v$,}\\
         	0, &\text{ otherwise. }\\	
         \end{cases}
\end{align}
\end{widetext}
\end{proof}
 
The reason that this proof works only for odd prime dimensions is that, if $m$ is not prime, then there will exist an $k$ and $a,b,c,e$ such that the LHS of~\eqref{eq:phase_cancel} is $1$ for conditions other than those of~\eqref{eq:phase_cancel}. 
Furthermore, when $m=2$, then there will be diagonal elements such that~\eqref{eq:phase_cancel} is $-1$ for $k=2$ and only for $k=4$ do we get actual recurrence (implying in turn that for $m=2$ the map is neither the dephasing map nor the identity map).
 
 However, in the following lemma, we show that for any odd dimension we can construct a unitary operator basis that does exhibit recurrence. 
 \begin{lemma}[Recurrence for odd dimension]
 Let $\dim{\mc H_S} = m^2, \dim{\mc H_R} = m$, where $m$ is odd. There exists a unitary $V$ acting on $\mc H_S \otimes \mc H_R$ such that
	\begin{align}
     \Tr_B(V^k (\rho \otimes \id_m )(V^\dagger)^k) = \pi^k_m(\rho).
	\end{align}
 	\end{lemma}
 
 	\begin{proof}
 	Consider the prime factor decomposition of $m= p_1 \dots p_l$. We can split the Hilbert spaces as 
 	\begin{align}
 	    \mc H_R \simeq \bigotimes_{j=1}^l \mc H_j,
 	\end{align}
 	where $\dim(\mc H_j) = p_j$. Moreover, let $A=\{\ket{\mathbf{r},\mathbf{s}}\}$ be an orthonormal basis of $\mc H_S$, where $\mathbf{r}, \mathbf{s} \in \mc S \coloneqq \times_{j=1}^l \{1, \dots, p_j\}$, so that $|\mc S| = m$. Now, we define the unitary
 	\begin{align}
 	    V &= \sum_{\mathbf{r},\mathbf{s} \in \mc S} \proj{\mathbf{r},\mathbf{s}}_S \otimes \left( \bigotimes_j U^{(j)}_{r_j,s_j} \right)_R,
 	\end{align} 
 	where $U^{(j)}_{r,s}$ acts non-trivially only on $\mc H_j$ and $r_j,s_j$ denote the $j$-th component of the respective strings. The result now follows in just the same way as in the previous proof, as
 	\begin{widetext}
 \begin{align}
     \Tr_B (V^k (\rho \otimes \id/m )(V^\dagger)^k) &= \sum_{\mathbf{r},\mathbf{s},\mathbf{u},\mathbf{v}} \proj{\mathbf{r},\mathbf{s}}\rho \proj{\mathbf{u},\mathbf{v}} \prod_j^l \left(\frac{1}{p_j}\Tr(U^{(j)}_{kr_j,ks_j} U^{(j)}_{-ku_j,-kv_j})\right)\\
     &= \sum_{\mathbf{r},\mathbf{s},\mathbf{u},\mathbf{v}} \proj{\mathbf{r},\mathbf{s}}\rho \proj{\mathbf{u},\mathbf{v}} \prod_j^l \theta_{p_j}(k,r_j,s_j,u_j,v_j) \\
     &= \pi^k_m(\rho), 
     \end{align}
     \end{widetext}
     since $k=m$ is by construction the smallest integer such that $k \text{ mod } p_j = 0$ for all $j$.
 \end{proof}
 
Also, it should be noted that the case of even dimension can also be considered very close to a perfect dephasing map: 
Within the cycle $k \in \{1, \dots, 2m\}$, the only two times at which the above map does not dephase perfectly is at $k=m$ and $k=2m$. 
At the latter, it yields the identity map, while at the former, it yields the identity map up to sign flips on a subset of its elements. 

\section{Entanglement-assisted private quantum channel} 
\label{app:pqc}

Here, we present the proofs for the ideal PQC presented in the main text and discuss its properties. 
As our construction does not fit into the usual formal framework of PQCs with classical keys, let us first specify in more detail what we mean by a private quantum channel with a quantum key.
We assume that Alice and Bob hold a shared quantum system $K=K_AK_B$ in a state vector $\ket{\Psi}_{K}$, which we refer to as the key, and that Alice wants to encode a quantum system $S$ with Hilbert-space $\mc H_S$. 
For notational simplicity we write $\mc H_{K_A} = \mc H_A$ and $\mc H_{K_B}=\mc H_B$. 
Then an ideal private quantum channel with key $\ket{\Psi}_K$ is given by a pair of quantum channels $\mc X: \mc S(\mc H_S\otimes \mc H_{A})\rightarrow \mc S(\mc H_S'\otimes \mc H_{A})$ and $\mc Y:\mc S(\mc H_S'\otimes \mc H_{B})\rightarrow \mc S(\mc H_S\otimes \mc H_{B})$ with the following properties.
First, there exists a fixed state $\tau$, such that for all auxiliary systems $E$ and all states $\rho_{SE}$ on $S$ and $E$ we have
\begin{align}\label{eq:quantumsecrecy}
	\tr_K \circ (\mc X\otimes \text{id}_{K_BE})(\rho_{SE} \otimes \proj{\Psi}_K) =\tau\otimes\rho_{E}.
\end{align}
This implies that an eavesdropper cannot learn anything from the encoded message, even when previously entangled with $S$. 
Second, the transmission is reliable, that is for all states $\rho$ on $S$ we have
\begin{align}
	\tr_K\circ (\mc Y \otimes \text{id}_{K_A})\circ (\mc X\otimes \text{id}_{K_B})(\rho\otimes \proj{\Psi}_K) = \rho. 
\end{align}
In the following, we will show that the construction sketched in the main text fulfills this definition and explore some of its additional properties.
We begin with the following Lemma.
	\begin{lemma}[Properties of a private quantum channel]
	  Let $\rho \in \mc S(\mc H_S)$ with $\dim(\mc H_S) = d$ and let $\ket{\phi^+} \in \mc H_K = \mc H_A \otimes \mc H_B$ be an $e$-dimensional, maximally entangled bipartite state vector with $e = (\lceil d^{1/2}\rceil)^2$. Then there exist unitaries $U\in \mc B(\mc H_S \otimes \mc H_A), V \in \mc B(\mc H_S \otimes \mc H_B)$ such that 
	  \begin{align} \label{eq:pinching_after_one}
	      \tr_{A,B}(U (\rho\otimes \proj{\phi^+})U^\dagger) = \id_d, \quad \forall \rho
	  \end{align}
	  and 
	  	  \begin{align} \label{eq:return_split}
	      V U (\rho\otimes \proj{\phi^+})U^\dagger V^\dagger = \rho \otimes \proj{\phi^+}, \quad \forall \rho.
	  \end{align}
	\end{lemma}

	\begin{proof}
	Consider first the case that $d$ is a square number, in which case $e=d$. We can assume w.l.o.g. that 
	\begin{align} \label{eq:splitting-entanglement}
	    \ket{\phi^+} = \ket{\phi_1^+} \otimes \ket{\phi_2^+}, 
	\end{align}
	where $\ket{\phi_i^+}$ are both $\sqrt{e}$-dimensional maximally entangled state vectors acting on $\mc H_{A_i} \otimes \mc H_{B_i}$ respectively, of the form 
	\begin{align}
	    \ket{\phi_i^+} = \frac{1}{e^{1/4}} \sum_{j=1}^{\sqrt{e}} \ket{j,j}_{A_iB_i}.
	\end{align}
	We can do this because Alice and Bob can always rotate between all maximally entangled states by applying local unitaries and hence prepare the above state. 
	We now define the unitaries
	\begin{align}
		U_I &= \sum_i^{d} \proj{i}_S \otimes (U_i)_{A_1}, \\
	     U_{J} &= \sum_j^{d} \proj{j}_S \otimes (U_j)_{A_2}, \\
	     U &= U_JU_I,
	 \end{align} 
	 where $\{U_i\}_{i=1}^{d}$, $\{U_j\}_{j=1}^{d}$ are unitary operator bases for $\mc H_{A_1}$ and $\mc H_{A_2}$ respectively, and $I = \{\ket{i}\}_{i=1}^d$ and $J = \{\ket{j}\}_{j=1}^d$ are any two 
	 \emph{mutually unbiased bases} (MUBs) for $\mc H_S$, that is, they are both orthonormal and 
\begin{align}
    |\braket{i}{j}|^2 = \frac{1}{d}, \quad \forall i,j.
\end{align}
In prime power dimension, there are known to exist sets of $d+1$ many of such MUBs, but there exist at least two in any dimension~\cite{Bengtsson_2017}. 

By direct evaluation, we now have
	  \begin{align}
	      &\tr_{A,B}(U (\rho\otimes \proj{\phi^+})U^\dagger)  \\
	     =& \sum_{i,i',j,j'} \ket{j}\braket{j}{i}\bra{i}\rho\ket{i'}\braket{i'}{j'}\bra{j'} \tr(U_iU_{i'}) \tr(U_jU_{j'})/d \\
	     =& \sum_{j} \tr(\rho) \frac{1}{d} \proj{j} = \id_d,
	  \end{align}
where we used both the orthonormality of the operator bases and the defining property of the MUBs. 

We now turn to the unitary $V$. The construction is very similar to that of $U$. In fact, we use the fact that, for any unitary $U$,
	\begin{align}
	    (U \otimes \bar{U} ) \ket{\phi_i^+} = \ket{\phi_i^+},
	\end{align}
	where the bar denotes complex conjugation. We therefore define
	\begin{align}
		V_I &= \sum_i^{d} \proj{i}_S \otimes (\bar{U}_i)_{B_1}, \\
	     V_J &= \sum_j^{d} \proj{j}_S \otimes (\bar{U}_j)_{B_2}, \\
	     V &= V_IV_J,
	 \end{align}
so that the unitaries now act on Bob's half of the entanglement.~\eqref{eq:return_split} then follows again by straightforward evaluation.

Finally, consider the case that $d$ is not a square number. $e$ is by construction always the smallest square number larger than, or equal to, $d$, so that we can always perform the splitting in~\eqref{eq:splitting-entanglement} in such a way that the resulting entangled states provide sufficient local randomness to perform the two dephasing operations.  
	\end{proof}

The above can now be used to construct an ideal PQC, as shown in the following.
	\begin{lemma}[Ideal private quantum channels] \label{lem:our_pqc}
		With the notation from the previous lemma, the maps
	 \begin{align}
	     \mc X(\cdot) &\coloneqq U (\cdot )U^\dagger ,\\
	     \mc Y(\cdot) &\coloneqq V (\cdot )V^\dagger
	 \end{align}
 form an ideal private quantum channel with key $\ket{\Psi}_K=\ket{\phi^+}$.
	\end{lemma}

	\begin{proof}
	The ideal reliability of the above construction follows immediately from~\eqref{eq:return_split}. The ideal security follows from the fact that every map $\mathcal{R}$ with the property that it completely randomizes a given system, 
\begin{align}
	\mathcal{R}(\rho) = \id_d, \quad \forall \rho \in \mc S(\mc H_S),    
\end{align}
completely destroys all correlations that this system may have had with other systems~\cite{Hayden2004Randomizinga}, in the sense that, for any extension $\rho_{SE}$ of some $\rho$, 
\begin{align} \label{eq:destroy_correlations}
     \norm{(\mathcal{R} \otimes \text{id})\rho_{SE} - \id_{d} \otimes \rho_E}_1 = 0.
 \end{align} 
But since $\tr_K\circ \mc X$ has this property, by~\eqref{eq:pinching_after_one},~\eqref{eq:destroy_correlations} implies ideal security in the sense of~\eqref{eq:quantumsecrecy}.
		\end{proof}

We now turn to a discussion of the properties of the above PQC. To begin with, note that it is catalytic in the sense that, in the absence of eavesdropping the entanglement is, at the end, returned back in its original state. This follows from~\eqref{eq:return_split}. Especially since entanglement is commonly considered an expensive resource, this is a very appealing feature, even though it is not very robust, as we will discuss in the next section.  

Secondly, our PQC construction is optimal when considered as a noisy process, in the sense that is impossible to construct an ideal PQC with less entanglement than we do, provided the global evolution is unitary. As in the case of the lower bounds for the dephasing map, discussed in Appendix~\ref{sec:necessary}, we prove this optimality with respect to approximate PQCs, in order to show that our results are robust against slight deviations from an ideal PQC. To do so, we call, in analogy to the classical PQC, \eqref{eq:pqc_reliability}, a private quantum channel with key $\ket{\Psi}_K$ \emph{$\epsilon$-reliable}, if, instead of \eqref{eq:quantumsecrecy}, it satisfies
\begin{align} \label{eq:approx_quantumsecrecy}
    \sup_{\rho_{S,E} \: \in \: \mc S(\mc H_S \otimes \mc H_E)} \norm{ \tr_K \circ (\mc X\otimes \text{id}_{K_BE})(\rho_{SE} \otimes \proj{\Psi}_K) - \tau\otimes\rho_{E}}_1 \leq \epsilon.
\end{align}

	\begin{lemma}[]
	 Let $(\mc X, \mc Y)$ be an $\epsilon$-reliable private quantum channel with key $\ket{\Psi}_K$ for a quantum system of dimension $d$. If $\mc X$ is a unitary channel, then there exists an $\epsilon_{cr}$ such that, for all $\epsilon < \epsilon_{cr}$,
	 \begin{align}
		 \dim(\mc H_A) \geq \max\left\{4,d^{1-\epsilon}\epsilon^{\frac{\epsilon}{2}}\right\}.
	 \end{align}
	\end{lemma}
	\begin{proof}
		The proof is fully analogous to the discussion of the quantum case in Appendix~\ref{sec:necessary}. We therefore only give a sketch. We have that $\tr_{K_B}(\ket{\Psi}_K)= \id_{d_A}$. Hence, $\epsilon$-reliability together with the fact that $\mc X = U \cdot U^\dagger$ for some unitary operator $U$ implies that the encoding channel on $S$ is a quantum noisy operation $\mc E_Q^{d_A}$ as defined in \eqref{eq:def:quantumnoise}. This further implies that $\tau = \id_d$, since the von Neumann entropy is non-decreasing under noisy operations and the channel has to work for the input state $\id_d$. 
	We now bound $d_A$ by considering a specific transition. Let $\ket{\Psi}_{SE}$ be the maximally entangled state over $SE$, where we choose the extension $\mc H_E$ to be a copy of $\mc H_S$. For this particular transition, $\epsilon$-reliability of the channel implies that 
	\begin{align}
	    \norm{\mc E_Q^{d_A} \otimes \text{id}_E (\proj{\Psi}_{SE}) - \id_d \otimes \id_d }_1 \leq \epsilon.
	\end{align}
By Fannes' inequality, this implies 
\begin{align}
	S(\mc E_Q^{d_A} \otimes \text{id}_E (\proj{\Psi}_{SE})) \geq  \log d^2 + \epsilon \log(\epsilon/d^2).
\end{align}
We now consider the bipartition of the system $SEA$ into $SE$ and $A$. Using the Lieb-Araki inequality and following, from here on, exactly the same reasoning as that of Appendix~\ref{sec:necessary} below Eq.~\eqref{eq:ent_bound}, yields the desired bound. 
	\end{proof}

\subsection{Error correction, authentication, key recycling} 
\label{sec:discussion_}

As noted above, a particularly convenient feature of our PQC construction is that it is catalytic. 
This property implies that, in the absence of eavesdropping, the quantum key, can be fully recycled. 
However, it is of course the basic premise of cryptography that one is not guaranteed the absence of eavesdropping. 
It is therefore natural to ask how robust our PQC-implementation is to eavesdropping, by asking: Can Alice and Bob correct errors inflicted by an eavesdropper? 
How well can Alice and Bob check whether eavesdropping has occurred? How much of the key can Alice and Bob reuse in case they detect eavesdropping? 

In this section we show that Alice and Bob can use additional ebits to error-correct, authenticate efficiently and recycle part of the key even when eavesdropping occurs. 
The results of this section are mostly a translation of the arguments and techniques of Ref.\ \cite{Leung2002Quantum} applied to our protocol.

\subsubsection{Error correction} 
\label{ssub:error_correction}

We first turn to the question of error correction. 
Consider, for simplicity, the case that Alice and Bob want to transmit a pure  two-qubit state vector $\ket{\phi}$ along our PQC construction (i.e., the setting depicted in Fig.~\ref{fig:pqc}). 
Following the results in the previous section, $\ket{\phi}$ can be sent using two ebits in the Bell state vector
\begin{align}
    \ket{\Phi^+} = \frac{1}{\sqrt{2}} (\ket{0,0} +\ket{1,1})
\end{align}
as a key.
We consider the effect of any Pauli error $P_i \in \{\mathbbm{1}, X, Y, Z\}^{\otimes 2}$ that may have occurred during transmission of the data. 
The reason for this is that the most general effect of eavesdropping on the encoded state $\id_d = \tr_K\circ \mc X(\proj{\phi})$ that is sent between Alice and Bob can be described by a quantum channel $\mc E$ with decomposition
 \begin{align}
     \mc E(\rho') &= \sum_{i,j=0}^{15} e_{i,j} P_i \rho' P_j^\dagger.
 \end{align}
Hence, if there exists a measurement using \emph{local operations with classical communication
(LOCC)} that lets Alice and Bob perfectly distinguish between any two Pauli errors without destroying the state, then they can decorrelate the message from an eavesdropper and also error-correct the message~\cite{Leung2002Quantum}.

We now turn to show that there exist choices for the unitary operator basis and MUBs in the PQC of Lemma~\ref{lem:our_pqc} such that Alice and Bob can discriminate any two Pauli error without destroying the transmitted state. This possibility arises because Alice and Bob can choose the encoding in such a way that there exists a one to one 
correspondence between Pauli errors and the final state of the entanglement they used for transmission. For this correspondence to arise it suffices to (a) use the unitary operator basis defined in~\eqref{eq:def_special_uni} as bases $\{U_i\}$ and $\{U_j\}$ in the construction of the unitaries $U$ and $V$; (b) choose $I = \{\ket{0}, \ket{1}\}$ and $J = \{\ket{+}= H\ket{0}, \ket{-}=H\ket{1}\}$, where $H$ is the Hadamard gate. For these choices, the total transmission process is given by  Fig.\ref{circ:qpqc}, as a circuit diagram. Here, possible errors are given by the dashed box, with Alice's encoding to the left and Bob's decoding to the right of the dashed box and where we ignore global phases (for example, identifying $Y \equiv XZ$) since they do not alter the outcome.
\begin{figure}[!t]
\includegraphics[width=0.48\textwidth]{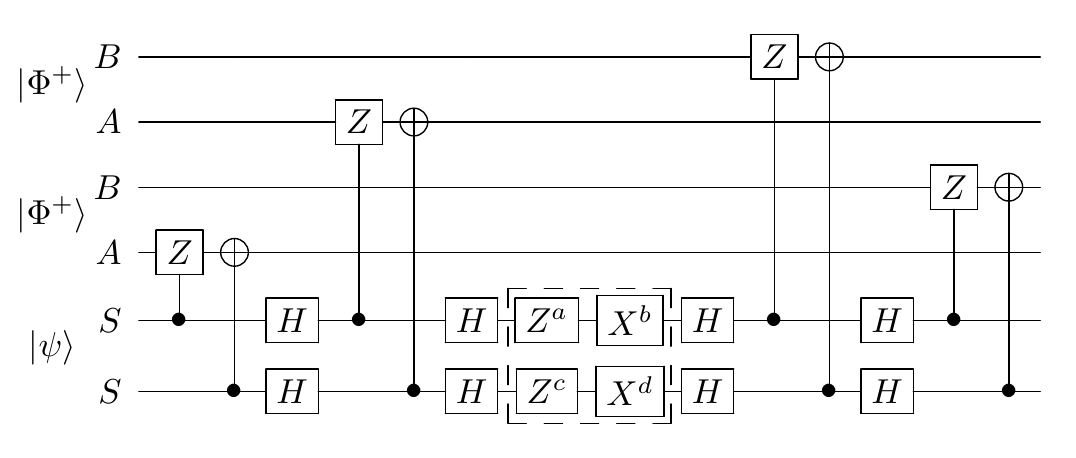}
\caption{The full entanglement-assisted PQC for a two-qubit state with Pauli matrices chosen as unitary operator basis and dephasing in the computational and Pauli $X$ eigenbases.}
\label{circ:qpqc}
\end{figure}

Using the relations
\begin{equation*}
\Qcircuit @C=.8em @R=0.1em @!R {
  & \targ & \gate{Z^a} & \gate{X^b} & \targ & \qw &&& 
    \gate{Z^a} & \gate{X^{b+d}} & \qw \\
   &&&&&& \push{\rule{.3em}{0em}=\rule{.3em}{0em}} &&&& \\
  & \ctrl{-2} & \gate{Z^c} & \gate{X^d} & \ctrl{-2} & \qw & & & \gate{Z^{a+c}} & \gate{X^d} &\qw
}
\end{equation*}
and
\begin{equation*}
\Qcircuit @C=.8em @R=0.1em @!R {
  & \gate{Z} & \gate{Z^a} & \gate{X^b} & \gate{Z} & \qw &&& 
    \gate{Z^{a+d}} & \gate{X^b} & \qw \\
   &&&&&& \push{\rule{.3em}{0em}=\rule{.3em}{0em}(-1)^{bd}} &&&& \\
  & \ctrl{-2} & \gate{Z^c} & \gate{X^d} & \ctrl{-2} & \qw & & & \gate{Z^{b+c}} & \gate{X^d} &\qw
}
\end{equation*}
together with the properties of entangled states, we find that, a Pauli error $P(a,b,c,d)$ described by the tuple $(a,b,c,d) \in \{0,1\}^{\times 4}$ yields the final state vector
\begin{align}
	(X^c Z^a)_{A_1} \otimes (X^{a+d}Z^b)_{A_2} \ket{\Phi^+}_{A_1B_1} \ket{\Phi^+}_{A_2B_2} P(a,b,c,d)\ket{\psi}, 
\end{align}
ignoring global phases and omitting identity operators. 
This implies that we can identify the tuple $(a,b,c,d)$ exactly just by distinguishing the Bell states, since no two different Pauli errors produce the same pair of Bell states, establishing the required correspondence. 
The same holds true also for mixed state messages, by linearity of quantum mechanics, and it also straightforwardly generalizes to the case of larger messages, since we can think of such messages as being sent in chunks of size $2$ using the above procedure.  

Going back to the case $n=2$, the above establishes a one to one correspondence between the sixteen possible Pauli errors on the ciphertext and the 16 possible combinations of Bell state vectors
\begin{align} \label{eq:bell_states}
    \ket{\Phi^\pm} &= \frac{1}{\sqrt{2}}(\ket{0,0} \pm \ket{1,1}), \\
    \ket{\Psi^\pm} &= \frac{1}{\sqrt{2}}(\ket{0,1} \pm \ket{1,0}).
\end{align}
If Alice and Bob could discriminate between these 16 combinations using LOCC measurements,
 then by the above this would mean that they can both decorrelate the decoded state from an eavesdropper as well as perform error-correction. 
However, this is not possible without the help of additional entanglement, since it is already impossible to discriminate between the four Bell state of a single ebit using LOCC measurements without further resources~\cite{Walgate2000Local}. 
However, the situation is different if Alice and Bob have access to additional ebits. 
In particular, let $\ket{\chi} \in \{\ket{\Phi^\pm}, \ket{\Psi^\pm}\}$ be an unknown Bell state vector. 
Then, if Alice and Bob share an auxiliary ebit prepared in the state vector $\ket{\Phi^+}$, they can each apply a $\mathsf{CNOT}$ gate, controlling on the auxilirary sytem $A$ and targeting $S$, which has the effect 
\begin{equation}
\begin{split}
    \ket{\Phi^+}_A \ket{\Phi^+}_S \to \ket{\Phi^+}_A \ket{\Phi^+}_S ,\\
    \ket{\Phi^+}_A \ket{\Phi^-}_S \to \ket{\Phi^-}_A \ket{\Phi^-}_S ,\\
    \ket{\Phi^+}_A \ket{\Psi^+}_S \to \ket{\Phi^+}_A \ket{\Psi^+}_S ,\\
    \ket{\Phi^+}_A \ket{\Psi^-}_S \to \ket{\Phi^-}_A \ket{\Psi^-}_S .\\
\end{split}
\end{equation}
If Alice and Bob now each measure their share of $A$ in the Pauli $X$ basis and their share of $S$ in the Pauli $Z$ basis and broadcast their measurement results, they can perfectly identify $\ket{\chi}$. 
Using this procedure for both ebits, they can extract full information about the error on the system and correct accordingly.

In summary, we have shown that Alice and Bob can perfectly discriminate between any two Pauli errors inflicted on the ciphertext during transmission, with the help of additional $n$ ebits. In this way, however, our PQC construction loses the advantage in resources over that of Ref.\ \cite{Leung2002Quantum}, where error correction is also possible using $2n$ ebits in total.

\subsubsection{Authentication and Key recycling} 
\label{ssub:key_recycling}

The above error correcting procedure has two disadvantages: Firstly, it requires a doubling of the total entanglement and, secondly, all the entanglement gets destroyed in the process. A more resource-effective strategy of Alice and Bob is to attempt to check for the occurrence of eavesdropping, destroying as little entanglement as possible, and consequently repeat the sending of the message whilst re-using as much of the entanglement as possible. We now discuss such a strategy in the asymptotic case, that is, when Alice and Bob send an $n$-qubit quantum message $\rho_S$ using $n$ ebits, in the limit $n \to \infty$. 

Let $\mathbf{v}$ be a $2n$-bit string encoding the final state of the $n$ ebits, with 
\begin{equation*}
     \ket{\Phi^+} \to 00, \ket{\Psi^+} \to 01, \,\,\ket{\Phi^-} \to 10, \ket{\Psi^-} \to 11 
 \end{equation*}
and the first two bits corresponding to the first ebit, etc. In order to check for the occurrence of eavesdropping, Alice and Bob can employ a LOCC-protocol constructed in Ref.\ \cite{Bennett1996MixedState} that yields the parity of any substring in $\mathbf{v}$, by destroying only a single ebit. Applying this protocol to $r$ random substrings of $\mathbf{v}$, one has
\begin{equation*}
    \text{Prob}(\mathbf{v} \neq 00\dots00 | \text{even parity in all $r$ rounds}) = \frac{1}{2^{-r}}.
\end{equation*}
Since $\mathbf{v} = 00\dots00$ corresponds to the case in which no Pauli error occurred, this implies that in case Alice and Bob measure no odd parity, they know that the message has been successfully transferred and that they can reuse their ebits for future communication, with exponentially small probability of mistake and at the cost of vanishingly few ebits. Now, in case they detect odd parity for any of their $r$ rounds, Alice and Bob consider the transfer unsuccessful and attempt to recycle as many of their ebits as possible. This amounts to estimating $\mathbf{v}$ while destroying as few ebits as possible in the course of doing so. We can directly apply a key recycling procedure presented in 
Ref.\ \cite{Leung2002Quantum} to our construction to achieve an asymptotic key recycling rate of $(1-H(\delta))$, where $H$ is the binary Shannon entropy and $\delta>0$ is a security parameter. We refer the reader to 
Ref.\ \cite{Leung2002Quantum} for details. 

These results should be compared with key recycling rates for the case of classical keys. There, the achievable recycling rates depend strongly on whether the message to be sent is classical (see, e.g., Refs.\ \cite{Bennett2014Quantum,Damgard2005Quantum,Fehr2017}) or quantum (see, e.g., Refs.\ \cite{Oppenheim2005How,Portmann2017,Hayden2016Universal}) and also on the possible attack scenarios that are being considered (see Refs.\ \cite{Hayden2016Universal} and~\cite{Portmann2017} for a discussion). Overall, however, the recycling rates can be considerably higher than those obtained here, albeit with significantly more complicated authentication schemes. Improving the recycling rates in the case of quantum keys thus remains an interesting open problem.

 \section{Quantum expanders} \label{sec:expanders_app}
 
 In this section, we discuss efficient approximate pinching to the main diagonal of an $d$-dimensional quantum
 system, of suitable dimension $d$, and provide the proof of Theorem 3. 
The proof of this statement is rooted in insights into random walks on expander graphs, is
connected to properties of Wigner functions of discrete Weyl systems and makes use of 
basic properties of quantum channels. It start from and builds upon the construction presented
in 
Ref.~\cite{Margulis}, which in turn derives from the classical description in Ref.~\
\cite{MargulisExpander}. The latter work discusses a random walk on an \emph{expander graph} 
featuring
the vertex set $\ZZ_e^2$, so an $e\times e$ integer lattice. Ref.~\cite{MargulisExpander} continues to show
that the random walk it constructs converges exponentially quickly to the 
uniform distribution ${\one}_{\ZZ_e^2}$ on this vertex set. Specifically, it is shown that
there exists a doubly stochastic matrix such that 
for any probability distribution ${P}$ on $\ZZ_e^2$, one has
\begin{equation}
	\|S^k (P) - {\one}_{\ZZ_e^2}\|_2\leq  \frac{5 \sqrt{2}}{8} \|S^{k-1} (P) - {\one}_{\ZZ_e^2} \|_2
\end{equation}
for $k\geq 1$ being an integer. Here, the action of the doubly stochastic map acting upon a 
distribution on ${\ZZ_e^2}$
 is written as $S(P)$.
On ${v}= ({v}_p,{v}_q)^T\in \ZZ_e^2$, this doubly stochastic matrix originates from random affine
transformations, drawn uniformly from the following eight  transformations
\begin{equation}
	v\mapsto T_1 {v},\, v\mapsto T_2 {v},\label{e1}\, 
	v\mapsto T_1 {v}+ {e}_1,\, {v}\mapsto T_2 {v}+{e}_2,
\end{equation}
and the four inverse transformations,
with 
\begin{eqnarray}\label{e3}
	T_1\coloneqq  \left(
	\begin{array}{cc}
	1& 2\\
	0 & 1\\
	\end{array}
	\right),\, 
	T_2\coloneqq  \left(
	\begin{array}{cc}
	1& 0\\
	2 & 1\\
	\end{array}
	\right),\, 
\end{eqnarray}
and
\begin{equation}\label{e4}
	{e}_1 = \left(
	\begin{array}{c}
	1 \\ 0
	\end{array}
	\right),\,
	{e}_2 = \left(
	\begin{array}{c}
	0 \\ 1
	\end{array}
	\right).
\end{equation}
The graph underlying this construction, with the $e\times e$ lattice as vertex set,
is an expander graph. Such an
expander graph is usually referred to as an
$(e^2,8,\lambda)$ expander graph with $\lambda \leq  5 \sqrt{2} /{8}$, in that it has $e^2$ vertices, 
each of which having $8$ neighbours in the graph. 
The matrix $S$ is sparse in that each row has $8$ entries only. Clearly, the above implies that
\begin{equation}\label{EXP}
	\|S^k (P) - \one_{\ZZ_e^2}\|_2\leq \sqrt{2} \left(
	 \frac{ 5 \sqrt{2}}{8}\right)^k.
\end{equation}
The prefactor of $\sqrt{2}$ originates from the fact that 
for any probability distribution $P$ on $\ZZ_e^2$, one has
\begin{eqnarray}
	\|P-  {\one}_{\ZZ_d^2}\|_2 &\leq & 
	((1-1/e^2) - (e^2-1)/e^2)^{1/2}\nonumber\\
	&=& \sqrt{2} (e^2-1)/e^2\nonumber\\
	&\leq &\sqrt{2}.
\end{eqnarray}	
We will relate this dimension $e$, which is left open at this point, to the physical dimension $d$ of the
quantum system subsequently.

The construction in a significantly altered setting will require some preparation.
For this, we turn to discussing the phase space $d\times d$ for the $d$-dimensional quantum system with odd $d$.
In the convention of Refs.~\cite{Margulis,DiscreteHudson},
for phase space coordinates $(p,q)\in \ZZ_d^2$, 
the discrete Wigner  function $W_M$ of an operator $M$ acting in Hilbert space can be written as
\begin{equation}
	W_M(p,q) = \frac{1}{d}{\rm tr}(w(p,q) \Pi w(p,q)^\dagger M),
\end{equation}
where $(p,q)\mapsto w(p,q)$ is the family of \emph{Weyl operators} and $\Pi$ is the 
\emph{parity operator}. The Weyl operators are composed of \emph{shift} and \emph{clock operators},
so the $X$ and $Z$ generalized Pauli matrices defined in Eq.\ (\ref{X}) and (\ref{Z}), respectively.
Any affine transformation $A$ the linear part of which having a unit determinant
on phase space coordinates $a\in \ZZ_d^2$, is unitarily reflected in Hilbert space as
\begin{equation}
	W_{U_A \rho U_A^\dagger} (a)= W_\rho (A^{-1} (a)).
\end{equation}
Wigner functions are normalized as
\begin{equation}\label{NORM}
	\sum_{(p,q)\in \ZZ_d^2} W_\rho(p,q)=1
\end{equation}	
for quantum states $\rho$. We treat Wigner functions for an operator $M$ as
matrices $W_M\in \CC^{d\times d}$, with real valued matrices for Hermitian $M$.
A first well-known insight is stated here as a separate lemma for completeness.

\begin{lemma}[Quantum states and Wigner functions]\label{StatesWigner} 
For two quantum states $\rho$ and $\sigma$ 
on a Hilbert space ${\cal H}_S$ of dimension $d$ associated with Wigner functions 
$W_\rho,W_\sigma: \ZZ_d^2\rightarrow \RR$, one has
\begin{equation}
	\|\rho-\sigma \|_2^2 =\| 
	W_\rho - W_\sigma\|_2^2=
	\sum_{{(p,q)\in \ZZ_d^2}}
	(W_\rho(p,q)- W_\sigma(p,q)  )^2.
\end{equation}
\end{lemma}
\begin{proof}
This statement follows directly from the property that the Hilbert Schmidt scalar product is inherited as
\begin{equation}
	\tr(\rho \sigma ) = \sum_{{(p,q)\in \ZZ_d^2}} W_\rho(p,q) W_\sigma(p,q),
\end{equation}
as follows from the analogous property of the characteristic function, and the definition of the 2-norm.
\end{proof}

The main insight of Ref.~\cite{Margulis} is to acknowledge that random walks on integer lattices
that are expander graphs can be connected to random unitary channels acting in Hilbert space
that inherit the mixing properties from the classical random walk, by resorting to a phase space picture.
The construction of Ref.~\cite{Margulis} builds upon the random walk on the 
Margulis expander graph~\cite{MargulisExpander} the vertex set of which is $\ZZ_e^2$ for some $e$ (here
taken to be different from $d$, as it will take a different role subsequently).
This random walk can be unitarily realized in quantum systems: In fact, the random walk 
follows directly from a convergence of a Wigner function, a function
that shares all properties of a probability distribution, except being positive. 
Following the construction of the random walk on the expander graph,
the quantum Margulis expander can be seen as a random unitary map
\begin{equation}\label{RUC}
	\rho\mapsto \mc D(\rho) = \frac{1}{8} \sum_{j=1}^8 U_j \rho U_j^\dagger,
\end{equation}	
of Kraus rank 8 with suitable unitary $\{U_i\}$ 
with the property that 
\begin{equation}
	\|\mc D(\rho) - \id_e\|_2 \leq  \frac{5 \sqrt{2}}{8}  \|\rho-  \id_e\|_2.
\end{equation}
A second insight on discrete Wigner functions that we will make use of is the following.

\begin{lemma}[Wigner functions of pinched quantum states] For any quantum state $\rho$ on 
${\cal H}_S$ of dimension $d$, the Wigner function of $\pi(\rho)$ satisfies
\begin{equation}
	W_{\pi(\rho)} (p,q) = W_{\pi(\rho)}(p',q)
\end{equation}
for all $q,p,p' =1,\dots, d$.
\end{lemma}
\begin{proof}
This statement follows directly from the definition of Wigner functions.
\end{proof}

This means that Wigner functions of pinched states are constant along the first coordinate.
Prepared in this fashion, we can finally turn to the new construction.
This construction of a random unitary channel will deviate from this construction in a
significant way: We identify for each $q\in \ZZ_d$ for $d=e^2$
the entire line 
$\{(p,q) \in \ZZ_d^2\}$ 
of the 
$d\times d$-dimensional phase space as a vectorized $e\times e$ lattice, on which the above affine
maps act. The property of the unit determinant 
of the linear part in the affine mapping is preserved.
In fact, it will act in precisely the same way on each line simultaneously, by applying 
one of the $8$ affine transformations defined in Eqs.\ (\ref{e1})-(\ref{e4}). This gives rise to $8$
affine maps on $\ZZ_d^2$.
Acting on Wigner functions,
this process can again be realized as a random unitary channel
\begin{equation}
	\rho\mapsto {\cal T}(\rho)  =\frac{1}{8}\sum_{j=1}^8 V_j \rho V_j^\dagger,
\end{equation}
with unitaries $\{V_i\}$. 
Clearly, the entire Wigner function $W_\rho$ of a state is normalized according to 
Eq.\ (\ref{NORM}). We refer to
\begin{equation}
	x_q\coloneqq \sum_{p\in \ZZ_d} W_\rho(p,q)
\end{equation}	
as the weight of each column.
We will now discuss the convergence properties of the above random unitary channel. For an integer $k\geq 1$,
we have
\begin{eqnarray}
	\|{\cal T}^k(\rho)- \pi(\rho)\|_2^2  &=& 
	\|W_{{\cal T}^k(\rho)}- W_{\pi(\rho)}\|_2^2 \\
	&= & \sum_{q\in  \ZZ_d} x_q^2 \sum_{p\in  \ZZ_d}
	\left(\frac{W_{{\cal T}^k(\rho)}(p,q)}{x_q} - \frac{1}{d} \right)^2 ,\nonumber
\end{eqnarray}
treating each columns separately. Using $x_q \leq d$ for all $q$ 
and using a worst case bound for all $q$ gives
\begin{eqnarray}
	\|{\cal T}^k(\rho)- \pi(\rho)\|_2^2  	
	&\leq &{d^3} \sum_{p\in  \ZZ_d}
	\left(\frac{W_{{\cal T}^k(\rho)}(p,q)}{x_q} - \frac{1}{d} \right)^2,
\end{eqnarray}
and following 	Eq.\ (\ref{EXP}), one obtains
\begin{eqnarray}
	 \|{\cal T}^k(\rho)- \pi(\rho)\|_2^2  	
	 &\leq &2 {d^3} \left(
	 \frac{ 5 \sqrt{2}}{8}\right)^{2k}.
\end{eqnarray}
In this way, we arrive at the anticipated result, by embedding the random unitary system into an
explicit quantum model, in the nomenclature of the main text.
\clearpage
\end{document}